\documentclass[11pt, draftclsnofoot, onecolumn]{IEEEtran}
\usepackage{amsmath,amssymb}
\usepackage[dvips]{graphicx}
\usepackage{amsfonts}
\usepackage[mathscr]{eucal}
\usepackage{latexsym}
\usepackage{amsthm}
\usepackage{exscale}
\usepackage[mathscr]{eucal}
\usepackage{bm}
\usepackage[dvipsnames]{color}
\usepackage{cases}
\usepackage{color}
\usepackage{epsfig}
\usepackage{algorithm}
\usepackage{algorithmic}
\usepackage[verbose,nospace,sort]{cite}
\usepackage{tabularx}
\usepackage{multirow}
\usepackage{multicol}
\usepackage{balance}
\usepackage{caption}
\usepackage{subcaption}
\usepackage{setspace}

\graphicspath{{./figs/}}

\theoremstyle{definition}
\newtheorem{proposition}{Proposition}
\newtheorem{theorem}{Theorem}

\newtheorem{lemma}{Lemma}

\newcommand{\SINR}{\mathrm{SINR}}

\newcommand{\diag}{\mathrm{diag}}

\newcommand{\dic}{\mathrm{dic}}
\newcommand{\ULA}{\mathrm{ULA}}
\newcommand{\UPA}{\mathrm{UPA}}

\usepackage{amsmath,amssymb}
\usepackage[dvips]{graphicx}
\usepackage{amsfonts}
\usepackage[mathscr]{eucal}
\usepackage{latexsym}
\usepackage{amsthm}
\usepackage{exscale}
\usepackage[mathscr]{eucal}
\usepackage{bm}
\usepackage[dvipsnames]{color}
\usepackage{cases}
\usepackage{color}
\usepackage{epsfig}
\usepackage{algorithm}
\usepackage{algorithmic}
\usepackage[verbose,nospace,sort]{cite}
\usepackage{tabularx}
\usepackage{multirow}
\usepackage{multicol}
\usepackage{balance}
\usepackage{caption}
\usepackage{subcaption}
\usepackage{setspace}
\usepackage{url}

\begin{document}
\title{Capacity Characterization and Formation Optimization for Multi-User MIMO Communications with UAV Swarm}
\author{Yong Zeng, {\it Fellow,~IEEE}\\
\thanks{Y. Zeng is with the National Mobile Communications Research Laboratory, Southeast University, Nanjing 210096, China, and also with Purple Mountain Laboratories, Nanjing 211111, China (e-mail: \{yong\_zeng\}@seu.edu.cn). (\emph{Corresponding author: Yong Zeng.})}
}

\maketitle

\begin{abstract}
For a multi-user multiple-input multiple-output (MU-MIMO) wireless communication system, imagining that the locations of the users are now fully controllable, what is the maximum sum-capacity, and what are the corresponding optimal user locations? While these questions are irrelevant in conventional human-centric communications with random user mobility, they become critically important for emerging applications involving ground or aerial robots. This paper addresses these fundamental questions in the context of MU-MIMO communications with an unmanned aerial vehicle (UAV) swarm acting as the users. To this end, we first derive closed-form expressions for the sum-capacity of MU-MIMO UAV swarm communications. Our results reveal that, compared to conventional MU-MIMO systems, the additional degrees of freedom provided by the coordinated mobility of the UAV swarm yields substantial capacity enhancement. Specifically, when the base station (BS) is equipped with an $M$-element uniform linear array (ULA), the full spatial multiplexing gain and beamforming gain, both equal to $M$, can be achieved simultaneously. For a BS with a uniform planar array (UPA), we show that asymptotically $\frac{\pi M}{4}$ users can simultaneously enjoy the full beamforming gain $M$. Furthermore, we propose a novel framework to optimize UAV swarm formation for maximizing the sum-capacity achieved by successive interference cancellation (SIC) and maximizing the sum-rate via treating interference as noise (TIN), taking into account practical considerations such as collision avoidance and swarm cohesion constraints. By exploiting the manifold structure of the array response vectors with respect to UAV directions, we develop an efficient algorithm to solve the resulting non-convex formation optimization problems. Extensive simulation results demonstrate that the proposed algorithms achieve near-optimal performance.
\end{abstract}


\section{Introduction}
Multi-user multiple-input and multiple-output (MU-MIMO) has been a key technology in modern wireless communication systems. It enables a single base station (BS) or access point (AP) equipped with multiple antennas to serve multiple user equipments (UEs) simultaneously over the same time-frequency resources \cite{377}. By leveraging the spatial degrees of freedom afforded by multiple antennas at the BS, MU-MIMO separates signals intended for (or received from) different UEs through their distinct spatial channel signatures. Compared to single-user MIMO (SU-MIMO), MU-MIMO provides greater flexibility in increasing spatial multiplexing gain, improving link reliability, and boosting spectral efficiency. Consequently, MU-MIMO had been widely used in WiFi and 4G LTE networks. In the 5G era, MU-MIMO has evolved in tandem with massive MIMO technology \cite{373}. Specifically, as the number of BS antennas grows significantly, far exceeding the number of active UEs, the channels corresponding to different UEs become asymptotically orthogonal under the so-called ``favorable propagation'' conditions, e.g., uncorrelated Rayleigh fading \cite{459}. In such regimes, even simple linear precoding or receive beamforming techniques can achieve near-optimal performance.

The capacity characterization of MU-MIMO represents a non-trivial generalization of the SU-MIMO case \cite{37,36}. In multi-user settings, the notion of capacity must be extended to a capacity region, which describes the complete set of achievable rate tuples for all UEs simultaneously. For both the uplink and downlink MU-MIMO scenarios, corresponding to  multiple-access channel (MAC) and broadcast channel (BC) \cite{212,263} respectively, the capacity regions are well established. Of particular interest is the sum-capacity. With circularly symmetric additive white Gaussian noise (AWGN), for any given channel realization, the sum-capacity of MAC can be expressed as a log-determinant function of the received signal covariance matrix. This sum-capacity is achieved when all UEs transmit at their maximum power and the BS employs optimal multi-user detection (MUD), typically realized via successive interference cancellation (SIC). In the downlink MU-MIMO setting, a parallel result holds. By leveraging uplink-downlink duality \cite{310,453,312,271}, the BC capacity  can be derived from the MAC solution, and the capacity-achieving strategy is known to be dirty-paper coding (DPC) \cite{5003,212}. However, both SIC and DPC are nonlinear techniques that entail high computational complexity and raise practical concerns such as error propagation and privacy. Consequently, a substantial body of research has focused on suboptimal but practical linear schemes, by treating  interference as noise (TIN) and the achievable communication rate is optimized through carefully designed linear transmit and receive beamformers \cite{182,226,5004}.

However, the aforementioned works on MU-MIMO capacity characterization and rate optimization all assume that UE locations are exogenous and uncontrollable. In such settings, UE-BS channels are determined by the environment, lying outside the scope of system optimization. This assumption is well-founded in conventional human-centric wireless communications, where user mobility is inherently random. Conversely, as wireless networks evolve from connecting people to connecting intelligence and everything, a new class of scenarios is emerging, where the system exerts full control over the physical positions of UEs. This is evident in network-connected autonomous agents, such as swarms of ground robots or unmanned aerial vehicles (UAVs) \cite{2005}. The potential usage scenarios include UAV swarm-based aerial inspection, relaying, and sensing, as illustrated in Fig.~\ref{F:UAVSwarmApplications}. In this paradigm, capacity and achievable data rates are expected to improve significantly, owing to the additional degrees of freedom offered by optimizing the locations or formation of the swarm UEs, which directly govern the UE-BS channels. Consequently, this new paradigm raises several key questions:

(1) What is the maximum sum-capacity achievable via SIC?

(2) What is the maximum sum-rate achievable via TIN?

(3) What constitutes the optimal swarm formation?

\begin{figure}
\centering
\includegraphics[scale=0.5]{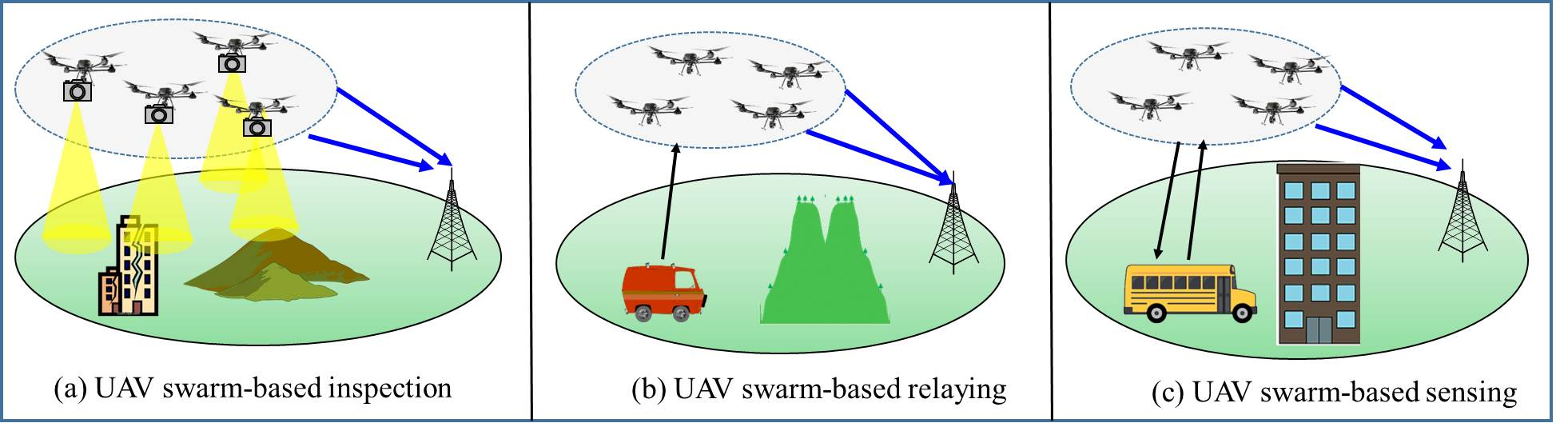}
\caption{Example usage scenarios of network-connected UAV swarm, where the UAVs in the swarm communicate with the BS with an additional design freedom via formation optimization.}\label{F:UAVSwarmApplications}
\end{figure}

This paper aims to address the above fundamental questions. In fact, exploiting node mobility to enhance communication capacity or data rates has been explored in prior work under different settings. A representative example is mobile relaying-assisted wireless communication, which gave rise to the data ferrying paradigm for delay-tolerant networks  \cite{638,653}. In this approach, a mobile relay first approaches the source node to collect data and subsequently moves toward the destination to offload it. This concept has since evolved into a more general framework of joint communication and trajectory optimization, particularly in the context of UAV-assisted communications \cite{641,1095}. Here, a UAV acting as a flying relay, BS, or AP can optimize its trajectory to maximize metrics such as throughput or energy efficiency \cite{904}. However, these works typically assume that the mobile nodes (e.g., UAVs) serve as dedicated helpers to provide communication services. The resulting performance gains over static deployments are primarily attributed to mobility-induced variations in path loss: improving link quality by flying closer to intended users and farther from interferers. Notably, achieving significant gains often requires the UAV to traverse substantial distances, say on the order of tens or even hundreds  of meters.

Besides mobility-induced path loss variation, node location optimization also introduces controllable phase variations across wireless links, where even relatively small changes in position can lead to significant differences in system performance. This principle was previously explored under the concept of line-of-sight (LoS) MIMO, wherein the positions of the transmitter and/or receiver antenna arrays were deliberately optimized to achieve a spatial multiplexing gain greater than one even in pure LoS propagation environments \cite{624}. A key insight from this line of work is that, depending on the link distance, the required antenna separation must be much larger than the conventional half-wavelength spacing to induce sufficient channel decorrelation via near-field or spherical-wave effects \cite{5005}. However, such large apertures are difficult to implement in practice, especially for compact user devices. More recently, the concepts of fluid antenna system (FAS) and movable antenna (MA) have been proposed to harness spatial diversity through physical repositioning of antenna elements \cite{5006,5007}. Nevertheless, their practical deployment faces challenges related to mechanical hardware complexity, actuation latency, and control overhead.

Different from the prior works, in this paper, we exploit the controllable UE mobility for capacity and rate optimization under the context of MU-MIMO with UAV swarm being the UEs. This scenario differs fundamentally from prior UAV-assisted communication paradigms \cite{649}, where UAVs primarily serve as aerial platform providing communication services for UEs. In contrast, here the UAVs themselves act as independent UEs, each transmitting or receiving its own data stream.
Crucially, optimizing UAV positions offers dual benefits: not only does it enable favorable path-loss characteristics (e.g., proximity to the BS), but it also allows deliberate control over inter-UAV channel correlations, which are inherently tied to the relative signal phases induced by their spatial configuration.
Moreover, unlike LoS MIMO or FAS/MA systems, which require specialized hardware to reposition antennas, the UAVs in our setup possess intrinsic high mobility by design, eliminating the need for additional mechanical or electrical actuation mechanisms. Finally, because the UAVs operate as independent transceivers rather than forming a synchronized distributed antenna array, the system avoids the stringent requirements on inter-node coordination, phase synchronization, and joint signal processing. This significantly reduces implementation complexity and enhances practical feasibility.

MIMO UAV swarm communications have attracted growing attentions recently. For example,  UAV swarm was exploited to form virtual or movable antenna array that cooperatively transmit/receive  data streams \cite{5009,5011}. Game-theoretical approach has been proposed  for decentralized 3D deployment for massive MIMO UAV swarm communications \cite{5008}. In \cite{5010}, gradient descent based algorithm was proposed to optimize the UAV swarm positions to achieve a high multiplexing gain in LoS MIMO backhaul. Despite this increasing research interest, the three fundamental questions outlined above regarding the maximum sum capacity,  achievable rate, and optimal swarm formation remain open. This gap motivates our current work.

The capacity characterization and formation optimization of MU-MIMO UAV swarm communications faces two major difficulties. The first one is finding a tractable relation between the UAV swarm locations  and the MU-MIMO channels. This issue is alleviated under the context of MIMO UAV  swarm communications, since compared to ground UEs, aerial UEs have elevated altitude, which make it more likely to have clear LoS link with the BS \cite{1095}. In this case, the LoS wireless channel depends on the UAV locations deterministically, via distance-dependent path loss and signal phase. The second major difficulty is that as the UAV locations  usually affect the channels in rather complicated manner, say a wavelength-scale variation of UE location would change the channel phase by $2\pi$, conventional convex or non-convex optimization techniques trying to optimize swarm locations directly fail. In this paper, we provide effective solutions to address such challenges. The main contributions of this paper are summarized as follows:
\begin{itemize}
\item {\it MU-MIMO UAV Swarm Capacity Characterization}: We derive the closed-form expressions for the sum-capacity of MU-MIMO communications with a UAV swarm being the UEs, by exploiting the new degree of freedom provided by the coordinated mobility of UAVs. In particular, for a BS equipped with $M$-element uniform linear array (ULA), theoretical analysis demonstrates that the sum-capacity can be expressed as $C=M\log_2(1+\rho M)$, where $\rho$ is the received signal-to-noise ratio (SNR) without considering beamforming gain. This result reveals that up to $M$ UAV UEs can be  simultaneously supported, each enjoying the full beamforming gain of $M$.  This finding stands in sharp contrast to conventional MU-MIMO communications, where full spatial multiplexing gain and full beamforming gain, both equal to $M$, typically cannot be achieved at the same time. For example, for massive MIMO systems under the so-called ``favorable propagation conditions'', the full beamforming gain  $M$ is achieved for all UEs only when the number of UEs $K$ satisfies $K\ll M$, in which case the asymptotical capacity is $C=K\log_2(1+\rho M)$.\footnote{Note that if we do not require full beamforming gain $M$ for all UEs, the sum-capacity can be higher by making full use of the spatial multiplexing gain $M$.} This is substantially smaller than that of the considered MIMO UAV swarm communications. The fundamental reasons for this enormous gain are twofold: 1) UE locations are fully controllable via UAV formation optimization; and 2) the BS-UAV channel is deterministically determined by the UAV locations under LoS scenarios. Consequently, up to $M$ UAVs can be deliberately positioned  along the $M$ orthogonal directions of the $M$-element ULA. In parallel with the notion of ``favorable propagation'' (e.g., uncorrelated Rayleigh fading) in massive MIMO setup, we term the free-space LoS channels as the ``ideal propagation``, as it simultaneously achieves full beamforming gain and full spatial multiplexing gain. Similar results hold when the BS  is equipped with $M$-element uniform planar array (UPA). A key difference is that, unlike the ULA case where exactly $M$ UEs achieve full beamforming gain, we show that asymptotically $\frac{\pi M}{4}$ UAV UEs can enjoy the full beamforming gain $M$ when $M$ is large. Building on these insights, we further derive closed-form capacity expressions for more general scenarios where the UAV swarm must be placed within prescribed distance and angular ranges.
\item {\it UAV Swarm Formation Optimization:} Taking into account practical constraints such as collision avoidance and swarm cohesion constraints, we propose a novel optimization framework to optimize UAV swarm formation for maximizing either the achievable sum-capacity under SIC or the sum-rate under TIN. The resulting optimization problems are highly non-convex, as UAV locations influence the BS-UE channels through both distance-dependent path loss and channel phase in a complicated manner. This renders standard convex or non-convex optimization techniques difficult to apply directly. Fortunately, by exploiting the manifold structure of the array response vectors with respect to UE directions, we develop an efficient algorithm to effectively solve the UAV swarm formation optimization problems. Extensive simulation results demonstrate that the proposed algorithms achieve near-optimal performance.
\end{itemize}

The rest of this paper is organized as follows. Section~\ref{sec:systemModel} introduces the system model. Section~\ref{sec:capacityCharacterization} presents  the capacity characterization of MU-MIMO communications with UAV swarm.  Section~\ref{sec:optimization} presents the UAV swarm formation optimization algorithms, by considering the practical constraints. Section~\ref{sec:simulation} presents the simulation results. Finally, we conclude the paper in Section~\ref{sec:Conclusion}.

{\it Notations: } Scalars are denoted by italic letters. Boldface lower- and upper-case letters denote vectors and matrices, respectively. $\mathbb{R}^{M\times 1 }$ and $\mathbb{C}^{M\times 1 }$ denote the space of $M$-dimensional real- and complex-valued vectors, respectively. For a vector $\mathbf a$, $\|\mathbf a\|$ represents its Euclidean norm, $\mathbf a^T$ and $\mathbf a^H$ denote its transpose and Hermitian transpose, respectively, and $\mathrm{diag}(\mathbf a)$ denotes a diagonal matrix with $\mathbf a$ being the diagonal elements. For a square matrix $\mathbf A$, $|\mathbf A|$ denotes its determinant. $\log_2(\cdot)$ denotes the  logarithm with base $2$. $\otimes$ denotes the Kronecker product.

\section{System Model}\label{sec:systemModel}
As shown in Fig.~\ref{F:systemModel}, we consider a MU-MIMO wireless communication system, where the BS equipped with $M$-element antenna array communicates with a UAV swarm consisting of $K$ single-antenna UAVs. Without loss of generality, we establish a Cartesian coordinate system so that the reference antenna  element (say the first element) of the BS  array locates at the origin.   Let $\mathbf q_k=r_k[\cos \theta_k \cos\phi_k, \cos\theta_k \sin \phi_k, \sin\theta_k]^T$ denote the 3D location of UAV $k$, $k=1,...,K$, where $r_k$ denotes the range,  $\theta_k\in [-\pi/2,\pi/2]$ and $\phi_k\in [-\pi/2,\pi/2]$ denote the elevation and azimuth angles, respectively. The set of all the $K$ UAV locations characterizes the  UAV swarm formation, denoted as $\mathbf Q\triangleq [\mathbf q_1,\cdots, \mathbf q_K]\in \mathbb{R}^{3\times K}$.

\begin{figure}
\centering
\includegraphics[scale=0.7]{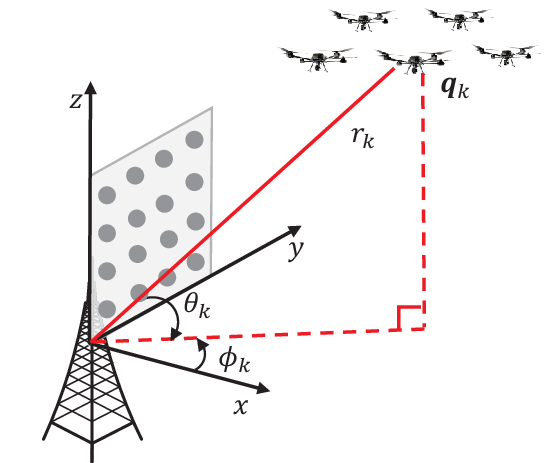}
\caption{MU-MIMO communication where a BS serves a UAV swarm, whose formation can be flexibly optimized.}\label{F:systemModel}
\end{figure}

For far-field free-space propagation, the LoS channel vector $\mathbf {\bar h}_k$ between the BS and UAV $k$ can be expressed as
\begin{align}\label{eq:hk}
\mathbf {\bar h}_k = \frac{\sqrt{\beta_0}e^{j\phi_k}}{r_{0}}\underbrace{\frac{r_{0}}{r_k} \mathbf a(\theta_k, \phi_k)}_{\triangleq \mathbf h_k(\mathbf q_k)}, k=1,\cdots, K,
\end{align}
where $\beta_0$ denotes the channel power gain at $1$ meter, $r_0$ is a reference range, $\phi_k$ is the phase shift of UAV $k$ to the reference antenna element of the BS array, and $\mathbf a(\theta_k, \phi_k)\in \mathbb{C}^{M\times 1}$ is the array response vector characterizing the relative phases between the $M$-element antenna array for a uniform plane wave (UPW) with propagation direction $(\theta_k,\phi_k)$. For convenience, in \eqref{eq:hk}, we have defined $\mathbf h_k(\mathbf q_k)$ to represent those channel components dependent on the  UAV location $\mathbf q_k$. Note that while the reference phase shift $\phi_k$ also depends on $\mathbf q_k$, as will become clear later, it does not affect the communication capacity or achievable rate.

{\it ULA:} If ULA is placed along the $y$ axis and assuming $\theta_k=0$ for simplicity, $\forall k$, then the array response vector is only a function of the azimuth AoA $\phi_k$, which can be written as
\begin{align}\label{eq:aULA}
\mathbf a_{\ULA}(\phi_k)=
\begin{bmatrix}
1 \\
e^{j\pi \sin\phi_k}\\
\vdots \\
e^{j\pi (M-1) \sin\phi_k}
\end{bmatrix},
\end{align}
where we have assumed that the adjacent elements are separated by half wavelength.

{\it UPA:} If UPA is placed in the y-z plane, let $M=M_yM_z$, with $M_y$ and $M_z$ denoting the number of elements along the y and z axis, respectively, then the array response vector can be expressed as
\begin{align}\label{eq:aUPA}
\mathbf a_{\UPA}(\theta_k, \phi_k)= \mathbf a_z(\theta_k)\otimes  \mathbf a_y(\theta_k, \phi_k),
\end{align}
where
\begin{align}\label{eq:azay}
\mathbf a_z(\theta_k)=
\begin{bmatrix}
1 \\
e^{j\pi \sin\theta_k}\\
\vdots \\
e^{j\pi (M_z-1) \sin\theta_k}
\end{bmatrix}, \
\mathbf a_y(\theta_k,\phi_k)=
\begin{bmatrix}
1 \\
e^{j\pi \cos\theta_k\sin\phi_k}\\
\vdots \\
e^{j\pi (M_y-1) \cos\theta_k\sin\phi_k}
\end{bmatrix}.
\end{align}
It is observed from \eqref{eq:hk}-\eqref{eq:azay} that the UAV location $\mathbf q_k$ affects the BS-UAV channel in two aspects: the channel amplitude is inversely proportional to the UAV range $r_k$ and the array steering vector $\mathbf a(\theta_k, \phi_k)$ only depends on the UAV direction $(\theta_k,\phi_k)$ while independent of $r_k$. 

Consider uplink communication, where the $K$ UAVs simultaneously communicate with the BS using the same time-frequency resource block. The received signal vector by the $M$-element antenna array of the BS is
\begin{align}\label{eq:y}
\mathbf y = \sum_{k=1}^K \mathbf {\bar h}_k \sqrt{P_k} s_k + \mathbf n,
\end{align}
where $P_k$ denotes the transmit power of UAV $k$, $s_k\sim \mathcal{CN} (0, 1)$ represents the circularly symmetric complex Gaussian (CSCG) information-bearing symbol, and $\mathbf n \in \mathcal{CN}(\mathbf 0, \sigma^2 \mathbf I_M)$ denotes the CSCG noise with power $\sigma^2$. The following two cases are considered:

{\it SIC:} The input-output relationship \eqref{eq:y} is the classic $K$-user MAC, whose sum-capacity can be achieved by SIC, given by
\begin{align}\label{eq:CQ}
C(\mathbf Q)&= \log_2 \left|\mathbf I + \frac{1}{\sigma^2}\sum_{k=1}^K \mathbf {\bar h}_k \mathbf {\bar h}_k^H P_k \right| \\
&=\log_2 \left|\mathbf I + \sum_{k=1}^K \mathbf h_k(\mathbf q_k) \mathbf h_k^H(\mathbf q_k) \frac{P_k\beta_0}{r_0^2\sigma^2} \right| \\
&= \log_2 \left|\mathbf I +  \mathbf{\bar{P}} \mathbf H^H(\mathbf Q) \mathbf H(\mathbf Q) \right|,\label{eq:sumCapacity}
\end{align}
where we define $\bar P_k\triangleq \frac{P_k\beta_0}{r_0^2 \sigma^2}$ to represent the received SNR for UAV $k$ at the reference distance $r_0$, $\mathbf H(\mathbf Q)=[\mathbf h_1(\mathbf q_1), \cdots \mathbf h_K(\mathbf q_K)]$, and $\mathbf {\bar{P}}=\diag\{\bar P_1, \cdots, \bar P_K\}$. Note that  the sum-capacity is expressed as a function of the UAV swarm formation $\mathbf Q$ in \eqref{eq:sumCapacity}, since the channel vectors $\mathbf H$ depend on $\mathbf Q$. Different from the conventional wireless communication systems, for MIMO UAV swarm communication where the formation $\mathbf Q$ is fully controllable, the maximal sum-capacity can be obtained by solving the following optimization problem:
\begin{align}\label{P:Cstar}
\begin{cases} \underset{\mathbf Q}{\max} & C(\mathbf Q)=\log_2\left|\mathbf I +  \mathbf{\bar{P}} \mathbf H^H(\mathbf Q) \mathbf H(\mathbf Q) \right| \\
 \text{s.t.} & \mathbf Q \in \mathcal Q,
 \end{cases}
\end{align}
where the mapping from UAV swarm $\mathbf Q$ to the channel vectors $\mathbf H(\mathbf Q)$ is specified by \eqref{eq:hk}-\eqref{eq:azay}, and $\mathcal Q$ is the feasible region for the UAV swarm formation that will be specified in Section~\ref{sec:optimalCapacity}.

{\it TIN:} Besides SIC, another common approach for MU-MIMO communication is to suppress the interference by applying linear receive beamforming while treating the residual interference as noise. Since such an approach in general cannot achieve the channel capacity, we usually use the achievable communication rate as the performance metric.  Let $\mathbf w_k\in \mathbb{C}^{M\times 1}$ denote the receive beamforming vector applied to the signal vector $\mathbf y$ in $\eqref{eq:y}$ for detecting UAV $k$'s information-bearing symbol $s_k$. The resulting signal is
\begin{align}
y_k=\mathbf w_k^H \mathbf y= \mathbf w_k^H \mathbf {\bar h}_k \sqrt{P_k} s_k + \sum_{j=1, j\neq k}^K \mathbf w_k^H \mathbf {\bar h}_j \sqrt{P_j}s_j + \mathbf w_k^H \mathbf n.
\end{align}
The corresponding signal-to-interference-and-noise ratio (SINR) is
\begin{align}\label{eq:SINRk}
\mathrm{SINR}_k(\mathbf Q, \mathbf w_k) = \frac{\bar{P}_k|\mathbf w_k^H \mathbf h_k(\mathbf q_k)|^2}{\sum_{j\neq k}\bar P_j|\mathbf w_k^H \mathbf h_j(\mathbf q_j)|^2+ \mathbf w_k^H \mathbf w_k}, \ k=1,\cdots, K,
\end{align}
where we have used the relationship between $\bar {\mathbf h}_k$ and $\mathbf h_k(\mathbf q_k)$ given in \eqref{eq:hk}, and  $\bar P_k$ is the reference SNR defined below \eqref{eq:sumCapacity}.
As indicated in \eqref{eq:SINRk}, the SINR for UAV $k$ is a function of the UAV swarm formation $\mathbf Q$ and its own receive beamforming vector $\mathbf w_k$.
The achievable communication sum rate by TIN is thus
\begin{align}
R(\mathbf Q, \mathbf W)=\sum_{k=1}^K \log_2\left(1+ \mathrm{SINR}_k(\mathbf Q, \mathbf w_k) \right),
\end{align}
where $\mathbf W\triangleq [\mathbf w_1, \cdots \mathbf w_K]$ denote the $K$ receive beamforming vectors. We may then formulate a sum rate maximization problem by jointly optimizing the receive beamforming vectors $\mathbf W$ and the UAV swarm formation $\mathbf Q$:
\begin{align}\label{P:Rstar}
\begin{cases} \underset{\mathbf Q, \mathbf W}{\max} & R(\mathbf Q, \mathbf W)= \sum_{k=1}^K  \log_2\left(1+ \mathrm{SINR}_k(\mathbf Q, \mathbf w_k) \right)\\
 \text{s.t.} & ~~ \mathbf Q \in \mathcal Q.
 \end{cases}
\end{align}
Different from the conventional MU-MIMO communication systems, the capacity and rate maximization problems \eqref{P:Cstar} and \eqref{P:Rstar} involve a new design variable, namely the UAV formation $\mathbf Q$. This thus offers a new dimension to enhance the communication capacity and achievable rate, thanks to the location-controllable UEs brought by agents like UAV swarms. However, directly solving  problems \eqref{P:Cstar} and \eqref{P:Rstar} is challenging, since the UAV formation $\mathbf Q$ affects the communication rate via channel vectors in a complicated manner, i.e., it affects both the amplitude and phase as indicated in \eqref{eq:hk}. Besides, the feasible region of the UAV swarm formation $\mathcal Q$ needs to be specified. As a result, standard convex or non-convex optimization techniques are difficult to be directly applied to solve the two problems. In the following, we first show that for some important scenarios, the optimal solution to optimization problems \eqref{P:Cstar} and \eqref{P:Rstar}  can be obtained in closed-form, based on which some important insights can be obtained. We then propose an efficient algorithm to solve the two problems for the general cases.

\section{MU-MIMO UAV Swarm Capacity Characterization}\label{sec:capacityCharacterization}
In this section, to gain some insights and characterize the communication capacity and achievable rate for MU-MIMO UAV swarm communication, we try to find the closed-form solutions for the capacity and rate maximization problems in \eqref{P:Cstar} and \eqref{P:Rstar}, respectively. To this end, we first derive a sufficient condition for the UAV swarm formation $\mathbf Q$ that achieves the global optimality. To that end, we define an alternative  optimization problem:
\begin{align}\label{P:CUB}
 \begin{cases} \underset{\mathbf Q}{\max} & \sum_{k=1}^K \log_2 \left(1+ \bar P_k \|\mathbf h_k(\mathbf Q)\|^2\right) \\
 \text{s.t.} & \mathbf Q \in \mathcal Q.
 \end{cases}
\end{align}
Note that problem \eqref{P:CUB} share the same feasible region $\mathcal Q$ as the capacity and rate maximization problems \eqref{P:Cstar} and \eqref{P:Rstar}, while they differ in their objective functions. In particular,  the objective function of \eqref{P:CUB} is given by the sum rate of the $K$ decoupled UEs with channel $\mathbf h_k$ as if all the remaining UEs were absent. In the following, we show that under some mild conditions, solving problem \eqref{P:CUB} is sufficient for solving the more challenging optimization problems \eqref{P:Cstar} and \eqref{P:Rstar}.

\subsection{Optimality Condition for Capacity Maximization}
We first reveal how problem \eqref{P:CUB} is related to the capacity maximization problem \eqref{P:Cstar}.

\begin{theorem}\label{theo:CUB}
Let $C_2^\star$ denote the optimal objective value for problem \eqref{P:CUB}. Then it provides an upper bound for problem \eqref{P:Cstar}, i.e.,
\begin{align}
C(\mathbf Q) \leq C_2^\star, \forall \mathbf Q\in \mathcal Q. \label{eq:CUB}
\end{align}
Furthermore, if and only if there exists an optimal solution $\mathbf Q_2^\star$ to problem \eqref{P:CUB} satisfying the following conditions:
\begin{align}
\mathbf h_k^H(\mathbf Q_2^\star) \mathbf h_j(\mathbf Q_2^\star) = 0, \forall j\neq k,  \label{eq:orthgogonal4}
\end{align}
then the upper bound \eqref{eq:CUB} is tight, and the equality is achieved by $\mathbf Q=\mathbf Q_2^\star$.
\end{theorem}
\begin{IEEEproof}
Please refer to Appendix~\ref{A:CUB}.
\end{IEEEproof}

\begin{lemma}\label{lemma1C}
Let $\mathbf Q_2^\star$ denote an optimal solution to problem \eqref{P:CUB}. If it satisfies the orthogonal conditions \eqref{eq:orthgogonal4}, then it must be an optimal solution to problem \eqref{P:Cstar} as well.
\end{lemma}
\begin{IEEEproof}
The result directly follows from Theorem~\ref{theo:CUB}, since under the conditions specified in Lemma~\ref{lemma1C}, $\mathbf Q_2^\star$ is also feasible to problem \eqref{P:Cstar} and it achieves the upper bound, then it must be an optimal solution.
\end{IEEEproof}

Lemma~\ref{lemma1C} provides an efficient method to optimally solve the challenging optimization problem \eqref{P:Cstar}, by firstly finding an optimal solution $\mathbf Q_2^\star$ to the relatively easier problem \eqref{P:CUB} so that the orthogonal conditions \eqref{eq:orthgogonal4} are satisfied. In this case, $\mathbf Q_2^\star$ is an optimal solution to problem \eqref{P:Cstar} as well. Note that if no such solution to \eqref{P:CUB} exists, \eqref{P:Cstar} remains challenging to solve, which will be addressed in Section~\ref{sec:optimization}.


\subsection{Optimality Condition for Rate Maximization}
In this subsection, for the case of TIN, we derive an optimality condition  for the rate maximization problem \eqref{P:Rstar}. Different from the capacity maximization problem \eqref{P:Cstar}, the rate maximization problem involves optimizing both receive beamforming matrix $\mathbf W$ and the UAV formation $\mathbf Q$. In the following, we first show that the problem can be reduced to optimizing $\mathbf Q$ alone.
\begin{theorem}\label{theo:1}
Solving the rate maximization problem \eqref{P:Rstar} is equivalent to solving the following swarm formation optimization problem:
\begin{align}\label{P:RQ}
 \begin{cases} \underset{\mathbf Q}{\max} & R(\mathbf Q) \triangleq \sum_{k=1}^K  \log_2\left(1+  \SINR_k(\mathbf Q) \right)\\
 \text{s.t.} & ~~ \mathbf Q \in \mathcal Q,
 \end{cases}
\end{align}
where
\begin{align}\label{eq:SINRkQ}
\SINR_k(\mathbf Q) = \bar P_k \mathbf h_k^H \Big(\sum_{j\neq k} \bar P_j \mathbf h_j\mathbf h_j^H + \mathbf I_M \Big)^{-1} \mathbf h_k.
\end{align}
\end{theorem}

\begin{IEEEproof}
Please refer to Appendix~\ref{A:theo:1}.
\end{IEEEproof}

Problem \eqref{P:RQ} is still difficult to be directly solved. Fortunately, similar to the capacity maximization problem discussed in the preceding subsection, it is also related to the optimization problem \eqref{P:CUB}.
\begin{theorem}\label{theo:2}
Let $C_{2}^\star$ be the optimal objective value to problem \eqref{P:CUB}. Then it provides an upper bound for the objective value of problem \eqref{P:RQ} , i.e.,
\begin{align}
R(\mathbf Q) \leq C_{2}^\star, \forall \mathbf Q\in \mathcal Q. \label{eq:RUB}
\end{align}
Besides, if and only if  there exists an optimal solution $\mathbf Q_2^\star$ to problem \eqref{P:CUB} satisfying the orthogonality condition \eqref{eq:orthgogonal4},
then the upper bound \eqref{eq:RUB} is tight, and equality is achieved by $\mathbf Q=\mathbf Q_2^\star$.
\end{theorem}
\begin{IEEEproof}
Please refer to Appendix~\ref{A:upperbound}.
\end{IEEEproof}

\begin{lemma}\label{lemma1}
Let $\mathbf Q_2^\star$ denote an optimal solution to problem \eqref{P:CUB}. If it satisfies the orthogonal conditions \eqref{eq:orthgogonal4}, then it must be an optimal solution to problem \eqref{P:RQ} as well.
\end{lemma}
\begin{IEEEproof}
The result directly follows from Theorem~\ref{theo:2}, since when the  conditions of Lemma~\ref{lemma1} are satisfied, $\mathbf Q_2^\star$ achieves the upper bound of the objective value of problem \eqref{P:RQ}, then it must be optimal.
\end{IEEEproof}

Therefore, similar to the capacity maximization problem, Lemma~\ref{lemma1} provides a method to optimally solve the challenging rate maximization problem \eqref{P:RQ}, by  finding an optimal solution $\mathbf Q_2^\star$ to the relatively easier problem \eqref{P:CUB} to satisfy the orthogonal conditions \eqref{eq:orthgogonal4}. In this case, $\mathbf Q_2^\star$ is also an optimal solution to problem \eqref{P:RQ}.

\subsection{Optimal Capacity Characterization}\label{sec:optimalCapacity}
Based on the above analysis, in order to find the optimal solution for both the capacity and rate maximization problems \eqref{P:Cstar} and \eqref{P:RQ}, a central task is to find the optimal solution $\mathbf Q_2^\star$ to the relatively easier problem \eqref{P:CUB}. If $\mathbf Q_2^\star$ satisfies the orthogonality condition \eqref{eq:orthgogonal4}, then not only SIC, but also TIN is capacity-achieving and they both share the same optimal solution $\mathbf Q_2^\star$. Therefore, in the following, we will focus on the optimization problem \eqref{P:CUB}.

Before solving problem \eqref{P:CUB}, let's specify some typical constraints for the feasible region $\mathcal Q$ of the swarm formation $\mathbf Q=[\mathbf q_1,\cdots, \mathbf q_K]$. Since $\mathbf q_k$ is fully determined by its range $r_k$ and direction $(\theta_k,\phi_k)$, the constraints can be either specified in terms of $\{\mathbf q_k\}$ or $\{r_k,\theta_k,\phi_k\}$:
\begin{itemize}
\item Range interval constraint:
\begin{align}\label{eq:C1}
\mathrm{C_1:}~ r_{\min}\leq r_k \leq r_{\max} , \forall k,
\end{align}
where $r_{\min}$ and $r_{\max}$ denote the minimum and maximum allowable ranges between the UAVs and the BS.
\item Angle interval constraint:
\begin{align}
\mathrm{C_2:}~ -\Theta \leq & \theta_k \leq \Theta, \forall k,\\
\mathrm{C_3:}~  -\Phi \leq & \phi_k \leq \Phi, \forall k,
\end{align}
where $\Theta\in [0,\pi/2]$ and $\Phi\in [0,\pi/2]$ denote the maximum allowable elevation and azimuth angles, respectively.
\item Collision avoidance constraint:
\begin{align}\label{eq:C4}
\mathrm{C_4:}~ \| \mathbf q_k - \mathbf q_j \| \geq d_{\min}, \forall j>k,
\end{align}
where $d_{\min}$ denotes the minimum separation between UAVs so as to avoid collision.
\item Swarm cohesion constraint:
\begin{align}
\mathrm{C_5:}~ \|\mathbf q_k-\mathbf q_j\| \leq d_{\max}, \forall j>k,\label{eq:C5}
\end{align}
where $d_{\max}\geq d_{\min}$ denotes the maximum separation between UAVs so as to ensure that all UAVs are within the swarm formation.
\end{itemize}
The feasible region of the UAV swarm formation $\mathcal Q$ can be specified as
\begin{align}
\mathcal Q=\left\{[\mathbf q_1,\cdots, \mathbf q_K]| \mathrm{C_1}-\mathrm{C_5} \right\},
\end{align}
where $\mathbf q_k=r_k[\cos \theta_k \cos\phi_k, \cos\theta_k \sin \phi_k, \sin\theta_k]^T$, $\forall k$. Depending on the practical scenarios, it is possible that only a subset of the constraints  $\mathrm{C_1}-\mathrm{C_5}$ are included in $\mathcal Q$. In the following, to get some insights, we only consider $\mathrm{C_1}-\mathrm{C_3}$ in $\mathcal Q$, while the other two constraints will be considered in the more general scenario discussed in Section~\ref{sec:optimization}.

Based on \eqref{eq:hk},  we have $\|\mathbf h_k(\mathbf q_k)\|^2=\frac{r_0^2M}{r_k^2}$, $\forall k$. Therefore, problem \eqref{P:CUB} can be equivalently expressed as
\begin{align}\label{P:CUBPolar}
 \begin{cases} \underset{\{r_k, \theta_k, \phi_k\}_{k=1}^K}{\max} & \sum_{k=1}^K \log_2 \left(1+ \frac{\bar P_kr_0^2M}{r_k^2}\right) \\
 \text{s.t.} &  \mathrm{C_1:}~ r_{\min}\leq r_k \leq r_{\max}, \forall k,\\
& \mathrm{C_2:}~ -\Theta \leq  \theta_k \leq \Theta, \forall k,\\
& \mathrm{C_3:}~  -\Phi \leq \phi_k \leq \Phi, \forall k.
 \end{cases}
\end{align}
Problem \eqref{P:CUBPolar} is trivial to solve, whose optimal solutions are $\{r_k=r_{\min},\theta_k,\phi_k\}_{k=1}^K$, with any $(\theta_k,\phi_k)$ satisfying $\mathrm C_2$ and $\mathrm C_3$. The fact that the objective value of problem \eqref{P:CUBPolar} is independent of the UAV direction $(\theta_k,\phi_k)$  provides us sufficient design freedom to find an optimal solution to problem \eqref{P:CUB} so that the orthogonal conditions \eqref{eq:orthgogonal4} are satisfied, as desired in Lemmas \ref{lemma1C} and \ref{lemma1}. With $\mathbf h_k$ given in \eqref{eq:hk}, the orthogonality conditions $\mathbf h_k^H \mathbf h_j=0$ are equivalent to $\mathbf a^H(\theta_k, \phi_k)  \mathbf a(\theta_j, \phi_j)=0$, $\forall j\neq k$. Therefore, the remaining task is to solve the following problem
\begin{align}\label{P:findAngle}
 \begin{cases} \mathrm{Find} & \{\theta_k, \phi_k\}_{k=1}^K \\
 \text{s.t.} & \mathrm{C_0:}~ \mathbf a^H(\theta_k, \phi_k)  \mathbf a(\theta_j, \phi_j)=0, \forall j\neq k,\\
& \mathrm{C_2:}~ -\Theta \leq  \theta_k \leq \Theta, \forall k,\\
& \mathrm{C_3:}~  -\Phi \leq \phi_k \leq \Phi, \forall k.
 \end{cases}
\end{align}
Problem \eqref{P:findAngle} is a feasibility problem, i.e., it only needs to find a set of variables  $\{\theta_k, \phi_k\}_{k=1}^K$ satisfying the specified constraints, without having to maximize any objective value. The key to solving problem \eqref{P:findAngle} is to satisfy the orthogonality constraint $\mathrm C_0$. To this end, we define the following function
\begin{align}\label{eq:fUPA}
f(\theta_k,\theta_j; \phi_k, \phi_j)&\triangleq |\mathbf a^H(\theta_k, \phi_k)  \mathbf a(\theta_j, \phi_j)|\\
&=\left|\left[\mathbf a_z(\theta_k)\otimes  \mathbf a_y(\theta_k, \phi_k)\right]^H\left[\mathbf a_z(\theta_j)\otimes  \mathbf a_y(\theta_j, \phi_j) \right]\right|\\
&=\left|\sum_{m_z=0}^{M_z-1} e^{j\pi m_z\left( \sin\theta_j-\sin\theta_k\right)} \right| \left|\sum_{m_y=0}^{M_y-1} e^{j\pi m_y \left( \cos\theta_j\sin\phi_j-\cos\theta_k\sin\phi_k\right)} \right|\\
&=\underbrace{\left|\frac{\sin\left(\frac{\pi M_z}{2}(\sin\theta_j-\sin\theta_k) \right)}{\sin\left(\frac{\pi }{2}(\sin\theta_j-\sin\theta_k) \right)} \right|}_{H_{M_z}(\Delta_{z,jk})}
\underbrace{\left|\frac{\sin\left(\frac{\pi M_y}{2}(\cos\theta_j\sin\phi_j-\cos\theta_k\sin\phi_k) \right)}{\sin\left(\frac{\pi }{2}(\cos\theta_j\sin\phi_j-\cos\theta_k\sin\phi_k \right)} \right|}_{H_{M_y}(\theta_k,\theta_j; \phi_k, \phi_j)},
\end{align}
where we have used the array response vector \eqref{eq:aUPA} for UPA, since it includes the ULA as a special case. Furthermore, we have defined the function
\begin{align}
H_{M}(\Delta)\triangleq \left|\frac{\sin\left(\frac{\pi M}{2}\Delta \right)}{\sin\left(\frac{\pi }{2}\Delta \right)} \right|,
\end{align}
and
\begin{align}
\Delta_{z,jk}\triangleq  \sin\theta_j-\sin\theta_k.
\end{align}
Note that  the function $f(\theta_k,\theta_j; \phi_k, \phi_j)$ is also known as the array beamforming pattern, which signifies the relative strength of the signal received from direction $(\theta_j, \phi_j)$, while the receive beamforming vector is designed to maximize the signal from a potentially different direction $(\theta_k, \phi_k)$.

 In the following, we first consider the special case of ULA, and then the more general UPA.

 \subsubsection{ULA}
 For the special case of ULA placed along the y-axis, we have $M_z=1$ and $M_y=M$. To gain the important insights, we let $ \theta_k=0$, $\forall k$. Then the remaining task is to find the set of azimuth angles $\{\phi_k\}_{k=1}^K$ satisfying $\mathrm{C}_0$ and $\mathrm{C}_3$. In this case, \eqref{eq:fUPA} reduces to
 \begin{align}
 f(0,0;\phi_k,\phi_j)=\underbrace{\left|\frac{\sin\left(\frac{\pi M}{2}(\sin\phi_j-\sin\phi_k) \right)}{\sin\left(\frac{\pi }{2}(\sin\phi_j-\sin\phi_k \right)} \right|}_{H_{M}(\Delta_{jk})},
 \end{align}
 where $\Delta_{jk}\triangleq \sin\phi_j-\sin\phi_k$ is also called the spatial frequency difference between UAV $j$ and $k$. Obviously, $|\Delta_{jk}|\leq 2$.

  Note that within the range $|\Delta_{jk}|\leq 2$, $H_M(\Delta_{jk})=0$ if and only if $\Delta_{jk} = \frac{2p}{M}$ and $|\Delta_{jk}|\neq 2$, $p=\pm 1, \pm 2....$. This implies that to ensure channel orthogonality, the spatial frequency difference between different UAVs should be integer multiples of the resolution $\frac{2}{M}$. With some abuse of notation, we use $\{\phi_l\}$ to denote the set of azimuth angles so that the channels are orthogonal. Therefore, we may let
  \begin{align}
  \bar \phi_l = \frac{2l}{M},\  l=0, \pm 1, \cdots, \pm L,
  \end{align}
  where $\bar \phi_l\triangleq \sin(\phi_l)\in [-\bar \Phi, \bar \Phi]$, with $\bar \Phi \triangleq \sin(\Phi)$, and $L$ is the parameter to be determined. In order to satisfy the constraint $\mathrm C_3$ so that $\bar \phi_l\leq \bar \Phi$, we should set $L=\left \lfloor \frac{M\bar \Phi}{2}\right \rfloor$. This ensures that the spatial frequency difference $\Delta$ of any pair of $\{\bar \phi_l\}$ is some integer multiples of $2/M$, which is a necessary condition to ensure orthogonality. However, we still need to ensure that $|\Delta|\neq 2$. Note that $|\Delta|= 2$  is possible if and only if $L=\frac{M}{2}$, or equivalently $\bar \Phi=1$ and $M$ is even. In this case, $\bar \phi_{-L}=-\bar \phi_L=-1$, and we can only retain either $\bar \phi_{-L}$ or $\bar \phi_{L}$, instead of both.

  Based on the above analysis, the set of UAV directions to ensure channel orthogonality are
  \begin{align}\label{eq:barPhil}
  \bar \phi_l = \frac{2l}{M}, \ \text{ where }
  l= \begin{cases} 0, \pm 1, \cdots, \pm \left(\frac{M}{2}-1\right), \frac{M}{2}, \  & \text{ if $\bar \Phi=1$ and $M$ is even} \\
  0, \pm 1, \cdots, \pm \left \lfloor \frac{M\bar \Phi}{2}\right \rfloor, \ & \text{ otherwise }.
  \end{cases}
  \end{align}
For both cases of \eqref{eq:barPhil}, the total number of orthogonal directions $N$ can be expressed in a unified form
\begin{align}\label{eq:NphiULA}
N(\bar \Phi)=\min \left \{M, \ 2 \left \lfloor \frac{M\bar \Phi}{2}\right \rfloor+1 \right \}.
\end{align}
It is not difficult to see that when $\bar \Phi=1$, $N(1)=M$ regardless whether $M$ is even or odd, i.e., the number of orthogonal directions equals to the number of ULA antenna elements $M$, as expected.

Based on the above analysis, we have the following Theorem:
 \begin{theorem}\label{theo:CULA}
 For $K$-user MIMO UAV swarm communication, if the BS is equipped with $M$-element ULA and the UAV swarm formation $\{\mathbf q_k\}_{k=1}^K$ can be freely optimized within the region  $(r_k, \bar{\phi}_k)\in [r_{\min}, r_{\max}] \times [-\bar{\Phi}, \bar{\Phi}]$, $\forall k$, where $K\leq N(\bar \Phi)$, then the optimal sum-capacity is
  \begin{align}\label{eq:CStar}
  C^\star = \sum_{k=1}^K \log_2\left(1+\frac{\bar P_kMr_0^2}{r_{\min}^2} \right),
  \end{align}
  which is achieved by $r_k=r_{\min}$, and $\{\bar{\phi}_k\}$ is any $K$-element subset of \eqref{eq:barPhil}.
 \end{theorem}
 \begin{IEEEproof}
 Based on the above derivations, the UAV formation specified by $(r_k,\theta_k=0,\phi_k)$ in Theorem~\ref{theo:CULA} is an optimal solution to problem \eqref{P:CUB} with the orthogonality conditions \eqref{eq:orthgogonal4} satisfied. Then based on Lemmas~\ref{lemma1C} and \ref{lemma1}, it is also the optimal solution to the capacity and rate maximization problems \eqref{P:Cstar} and \eqref{P:Rstar}. By substituting the solution, the maximum capacity \eqref{eq:CStar} can be directly obtained.
 \end{IEEEproof}
 If the number of UAVs in the swarm can be freely chosen as well, then we may take the maximum allowable values $K=N(\bar \Phi)$ to further maximize $C^\star$ in \eqref{eq:CStar}. Further assume that $\bar P_k=\bar P$, $\forall k$ for notational simplicity, we then have the following result:
  \begin{lemma}\label{lemma:CULA}
 For multi-user MIMO UAV swarm communication, if the BS is equipped with $M$-element ULA and the UAV swarm formation can be freely optimized within the region  $(r_k, \bar{\phi}_k)\in [r_{\min}, r_{\max}] \times [-\bar{\Phi}, \bar{\Phi}]$, $\forall k$, then the maximum sum-capacity is
  \begin{align}\label{eq:CStarMax0}
  C^\star =  \min \left \{M, \ 2 \left \lfloor \frac{M\bar \Phi}{2}\right \rfloor+1 \right \} \log_2\left(1+\rho M \right),
  \end{align}
 where $\rho\triangleq \frac{\bar{P}r_0^2}{r_{\min}^2}$ is the received SNR without considering the beamforming gain.
 \end{lemma}
 Furthermore, for the special case of $\bar \Phi=1$, then the optimal sum-capacity in \eqref{eq:CStarMax0} reduces to
   \begin{align}\label{eq:CStarMax}
  C^\star = M\log_2\left(1+\rho M \right).
  \end{align}
 The result \eqref{eq:CStarMax} gives a very simple and elegant capacity characterization for MU-MIMO communications with fully controllable UE locations like UAV swarms. It reveals that both the full spatial multiplexing gain and the beamforming gain $M$ are simultaneously achieved. This is remarkable since it is in a sharp contrast to the traditional MU-MIMO communication systems that usually suffer from a trade-off between spatial multiplexing gain and beamforming gain. As a concrete example, for the classic multi-user massive MIMO system under the so-called ``favorable propagation`` conditions, i.e., all users experience uncorrelated Rayleigh fading, then the inter-user channels become asymptotically orthogonal only when $M\gg K$, and the maximum sum-capacity is \cite{459}
 \begin{align}\label{eq:CMssiveMIMO}
 C^\star \approx K \log_2\left(1+ \rho M\right)\ll M \log_2\left(1+\rho M \right).
 \end{align}
 In this case, the full beamforming gain $M$ is achieved only for  $K\ll M$ UEs. Such an enormous gain by MIMO UAV swarm system is attributed by two factors: 1) UE locations are fully controllable via UAV formation optimization; and 2) the BS-UAV channel is deterministically determined by the UAV locations under LoS scenarios. Note that both systems achieve orthogonal inter-user channels, hence they both achieve the full beamforming gain $M$ with the very simple maximal-ratio combining (MRC) beamforming. However, for massive MIMO systems, with the assumption of uncorrelated Rayleigh fading channels, the orthogonality can only be  achieved asymptotically with a compromise of the number of supported UEs. By contrast, for the newly considered system with fully controllable UE locations under the deterministic LoS channels, orthogonality is achieved deterministically without sacrificing the spatial multiplexing gain. Thus, in parallel to  the ``favorable propagation'' in massive MIMO setup, we call the free-space LoS as the ``ideal propagation'' in the sense that it simultaneously achieves the full spatial multiplexing and beamforming gains when UE locations are fully controllable.

\subsubsection{UPA}
Next, we provide the capacity characterization for the UPA case. For $f(\theta_k,\theta_j; \phi_k, \phi_j)=0$ in \eqref{eq:fUPA} to be satisfied, we should have either $H_{M_z}(\Delta_{z,jk})=0$ or $H_{M_y}(\theta_k,\theta_j;\phi_k,\phi_j)=0$. For the former, the analysis in the previous subsection for ULA can be directly applied to determine the set of orthogonal elevation angles $\{\theta_l\}$. Similar to \eqref{eq:barPhil}, we should set
 \begin{align}\label{eq:barThetal}
  \bar \theta_l = \frac{2l}{M_z}, \ l \in \mathcal {L}, \ \text{ where }
  \mathcal L = \begin{cases}\left \{ 0, \pm 1,\cdots \pm \left(\frac{M_z}{2}-1\right), \frac{M_z}{2}\right\}, \  & \text{ if $\bar \Theta=1$ and $M_z$ is even} \\
  \left\{ 0, \pm 1, \cdots, \pm \left \lfloor \frac{M_z\bar \Theta}{2}\right \rfloor \right\}, \ & \text{ otherwise }.
  \end{cases}
  \end{align}
  For both cases in \eqref{eq:barThetal}, the total number of orthogonal elevation angles can be expressed in a unified form as
  \begin{align}\label{eq:Ntheta}
  N_{\theta}(\bar \Theta)\triangleq |\mathcal L|= \min\left\{M_z,\  2\left \lfloor \frac{M_z\bar \Theta}{2}\right \rfloor+1 \right\}.
  \end{align}

Next, for each specified elevation angle $\bar \theta_l$ in the set \eqref{eq:barThetal}, we aim to find the set of orthogonal azimuth angles $\{\bar \phi_p\}$ so that the second term $H_{M_y}(\theta_k,\theta_j; \phi_k, \phi_j)$ in \eqref{eq:fUPA} is $0$. In this case, we have
\begin{align}
H_{M_y}(\theta_k,\theta_j; \phi_k, \phi_j)_{|\theta_j=\theta_k=\theta_l}&=\left|\frac{\sin\left(\frac{\pi M_y}{2}\cos\theta_l(\sin\phi_j-\sin\phi_k) \right)}{\sin\left(\frac{\pi }{2}\cos\theta_l(\sin\phi_j-\sin\phi_k \right)} \right|\\
&=H_{M_y}\left(\Delta_{jk} \cos \theta_l \right),
\end{align}
where $\Delta_{jk}\triangleq \bar \phi_j - \bar \phi_k$. Within the range $|\bar \Delta_{jk}|\leq 2$, $H_{M_y}\left(\Delta_{jk} \cos \theta_l \right)=0$ if and only if $\Delta_{jk} \cos \theta_l =\frac{2p}{M_y}$, $p=\pm 1, \pm 2,\cdots$ and $|\Delta_{jk} \cos \theta_l |\neq 2$. Thus, for any given elevation angle $\theta_l$ so that $\cos \theta_l\neq 0$, the set of orthogonal azimuth angles can be found as
\begin{align}\label{eq:bar_phi_UPA}
\bar \phi_p = \frac{2p}{M_y\cos \theta_l}, \ p=0, \pm 1, \cdots \pm P_l,
\end{align}
where $P_l$ is a parameter to be determined. To ensure $\bar \phi_p\in [-\bar \Phi, \bar \Phi]$, we should let $P_l =\left \lfloor \frac{\bar \Phi M_y \cos\theta_l}{2}\right \rfloor$. Furthermore, we need to ensure that $|\Delta_{jk} \cos \theta_l |\neq 2$. Note that since $|\Delta_{jk}|\leq 2$ and $|\cos\theta_l|\leq 1$, the only case for $|\Delta \cos \theta_l |= 2$ is when $|\Delta_{jk}| = 2$ and $|\cos\theta_l|=1$. Furthermore, if $\exists j,k$ such that $|\Delta_{jk}|=2$, it implies that the minimum and maximum angles in \eqref{eq:bar_phi_UPA} are $\bar \phi_{-P_l}=-\bar \phi_{P_l}=-1$,  which is equivalent to $P_l =\left \lfloor \frac{\bar \Phi M_y}{2}\right \rfloor=\frac{\bar \Phi M_y}{2}=\frac{M_y}{2}$, or $\bar \Phi=1$ and $M_y$ is even. In this case, we can only retain either $\bar \phi_{-P_l}$ or $\bar \phi_{P_l}$, instead of both.

Based on the above derivations, for any given elevation angle $\theta_l$, the set of orthogonal azimuth angles are given by
\begin{align}\label{eq:bar_phi_UPA2}
\bar \phi_p = \frac{2p}{M_y\cos \theta_l}, \ \text{ where }
p=\begin{cases}
0, \pm 1, \cdots, \pm \left(\frac{M_y}{2}-1\right), \frac{M_y}{2} & \text{ if $\theta_l=0 $, $\bar{\Phi}=1$ and $M_y$ is even}\\
 0, \pm 1, \cdots, \pm \left \lfloor \frac{\bar \Phi M_y \cos\theta_l}{2}\right \rfloor & \text{ otherwise }
\end{cases}
\end{align}
Note that if $\theta_l=\frac{\pi}{2}$ or $\cos(\theta_l)=0$, the above expression is still applicable, though we have the technique issue of dividing by $0$. In this case, we may simply set the azimuth angle to $\bar \phi_0=0$. For all cases considered above, the total number of orthogonal azimuth angles for given elevation angle $\theta_l$ can be expressed in a unified form as
\begin{align}\label{eq:Nphi}
N_{\phi,l}(\bar \Phi) = \min \left\{M_y,\  2  \left \lfloor \frac{\bar \Phi M_y \cos\theta_l}{2}\right \rfloor +1\right\}, \ l \in \mathcal {L}.
\end{align}
As a result, for UPA, the total number of orthogonal directions is
\begin{align}\label{eq:NExact}
N_{\UPA}&=\sum_{l\in \mathcal L} N_{\phi,l}(\bar \Phi)\\
&= \sum_{l\in \mathcal L} \min \left\{M_y,\  2  \left \lfloor \frac{\bar \Phi M_y \cos\theta_l}{2}\right \rfloor +1\right\}.
\end{align}

In order to find the closed-form expression, we consider the scenario when both \eqref{eq:barThetal} and \eqref{eq:bar_phi_UPA2} take the second case, while it can be shown that the subsequent results are also applicable for other cases. Then we have
\begin{align}
N_{\UPA} 
& = \sum_{l=-\left \lfloor\frac{M_z\bar \Theta}{2} \right \rfloor}^{\left \lfloor\frac{M_z\bar \Theta}{2} \right \rfloor} \left( 2 \left \lfloor \frac{\bar \Phi M_y \cos\theta_l}{2}\right \rfloor +1\right)\\
& = 2\left \lfloor\frac{M_z\bar \Theta}{2} \right \rfloor+1 +2\sum_{l=-\left \lfloor\frac{M_z\bar \Theta}{2} \right \rfloor}^{\left \lfloor\frac{M_z\bar \Theta}{2} \right \rfloor} \left \lfloor \frac{\bar \Phi M_y \sqrt{1-\frac{4l^2}{M_z^2}}}{2}\right \rfloor,\label{eq:NUPA4}
\end{align}
where we have used the identity $\cos(\theta_l)=\sqrt{1-\bar \theta_l^2}=\sqrt{1-\frac{4l^2}{M_z^2}}$.

Finding the closed-form expression for $N_\UPA$ in \eqref{eq:NUPA4} is challenging. Fortunately, by deriving the upper and lower bounds of $N_\UPA$, the closed-form expression can be found asymptotically.
\begin{theorem}\label{theo:NUPABounds}
If $M_z\gg \frac{2}{\bar \Theta}$, the number of orthogonal directions $N_\UPA$ for an $M_y\times M_z$ UPA over the angular range $[-\Theta, \Theta]\times [-\Phi, \Phi]$ is lower- and upper-bounded by
\begin{align}\label{eq:NUPAbounds}
M_z\left(\frac{ \bar \Phi  M_y \left(\Theta+\bar \Theta \cos \Theta \right)}{2} - \bar \Theta \right)  +1\leq N_\UPA \leq \min\left\{ M_z\left(\frac{ \bar \Phi  M_y \left(\Theta+\bar \Theta \cos \Theta \right)}{2} +\bar \Theta \right)  +1, M_yM_z\right\}.
\end{align}
\end{theorem}
\begin{IEEEproof}
Please refer to Appendix~\ref{A:NUPABounds}.
\end{IEEEproof}

\begin{lemma}\label{lemma:NUPA2}
If $M_z\gg \frac{2}{\bar \Theta}$ and $M_y  \bar \Phi \left(\Theta+\bar \Theta \cos \Theta \right)\gg 2\bar \Theta$, then the number of orthogonal directions for $M=M_y \times M_z$ dimensional UPA over the angular range $\left [-\Theta, \Theta \right] \times \left [-\Phi, \Phi \right]$ is
\begin{align}
N_\UPA \approx \frac{\bar \Phi \left(\Theta+\bar \Theta \cos \Theta    \right)M}{2}.\label{eq:NUPA2}
\end{align}
\end{lemma}
\begin{IEEEproof}
Lemma~\ref{lemma:NUPA2} follows from the fact that under the specified condition, the lower and upper bounds of $N_\UPA$ specified in Theorem~\ref{theo:NUPABounds} converge to the same value in  \eqref{eq:NUPA2}. Note that we have used the fact that for all $0\leq \Phi\leq \pi/2$ and $0\leq \Theta \leq \pi/2$, $\frac{\bar \Phi \left(\Theta+\bar \Theta \cos \Theta \right)}{2}<1$, so that the upper bound $M_yM_z$ in \eqref{eq:NUPAbounds} will not be activated under the asymptotical condition specified in Lemma~\ref{lemma:NUPA2}. Furthermore, since $\bar \Theta \leq 1$, one sufficient condition for $M_y  \bar \Phi \left(\Theta+\bar \Theta \cos \Theta \right)\gg 2\bar \Theta$ is $M_y \bar\Phi \max\{\Theta, \cos\Theta\}\gg 2$, which is not difficult to achieve when $M_y\gg 1$.
\end{IEEEproof}
Lemma~\ref{lemma:NUPA2} shows that the asymptotical number of orthogonal directions for UPA increases linearly with  $M$ and $\bar \Phi$, while nonlinearly and with diminishing gain with respect to the elevation angle range $\bar \Theta$. This is expected since at larger elevation angles $\theta$, the projection of the signal direction in the azimuth plane reduces by a factor $\cos(\theta)$, this thus provides poorer azimuth resolution.

Similar to Lemma~\ref{lemma:CULA} for the ULA case, we immediately have the asymptotical capacity characterization for the UPA case.
  \begin{lemma}\label{lemma:CUPA}
 For MU-MIMO UAV swarm communication, if the BS is equipped with $M$-element UPA and the UAV swarm formation can be freely optimized within the region
  $(r_k, \bar{\theta}_k, \bar{\phi}_k)\in [r_{\min}, r_{\max}] \times [-\bar \Theta, \bar \Theta] \times [-\bar{\Phi}, \bar{\Phi}]$, $\forall k$, then the asymptotical sum-capacity ensuring full beamforming gain is
  \begin{align}\label{eq:CStarMaxUPA2}
  C^\star \approx  \frac{\bar \Phi \left(\Theta+\bar \Theta \cos \Theta \right)M }{2} \log_2\left(1+\rho M \right).
  \end{align}
 \end{lemma}

For the special case when $\bar \Theta=\bar \Phi=1$, \eqref{eq:CStarMaxUPA2} reduces to
  \begin{align}\label{eq:CStarMaxUPA3}
  C^\star \approx  \frac{\pi M}{4} \log_2\left(1+\rho M \right).
  \end{align}
Similar to the ULA case in \eqref{eq:CStarMax}, the expression \eqref{eq:CStarMaxUPA3} gives an elegant expression for the capacity of location-controllable UEs ensuring full beamforming gain $M$. Note that different from that for the ULA case, asymptotically $\frac{\pi M}{4}$ UEs would enjoy the full beamforming gain, which is slightly smaller than the full spatial multiplexing gain $M$. In fact, as long as the number of UEs $K\leq \pi M/4$, full beamforming and spatial multiplexing gains are simultaneously achieved. Similar to the ULA case, thanks to the location-controllable UE and the ideal LoS condition, the asymptotical capacity $C^\star$ in \eqref{eq:CStarMaxUPA3} is much larger than that of the classical massive MIMO setup $K \log_2\left(1+\rho M \right)$, which requires $K\ll M$.

In order to validate the bounds and asymptotical analysis presented in this subsection, Fig.~\ref{F:numdirectionsUPA} plots the number of orthogonal directions normalized by the total number of UPA elements, i.e., the ratio $N_{\UPA}/M$, where different curves correspond to various bounds or approximations for $N_{\UPA}$. Specifically, the curves labelled as ``Upper bound'' and ``Lower bound'' are the results in \eqref{eq:NUPAbounds}, ``Exact'' means the actual value based on \eqref{eq:NExact}, and ``Approx'' is the pre-log factor $\frac{\pi M}{4}$ in \eqref{eq:CStarMaxUPA3}. The number of antennas $M=M_yM_z$ increases by fixing $M_z=20$ while varying $M_y$. Fig.~\ref{F:numdirectionsUPA} shows that when $M$ is large, all the bounds and asymptotical closed-form expressions converge to the exact value, which validate our theoretical analysis.


\begin{figure}
\centering
\includegraphics[scale=0.6]{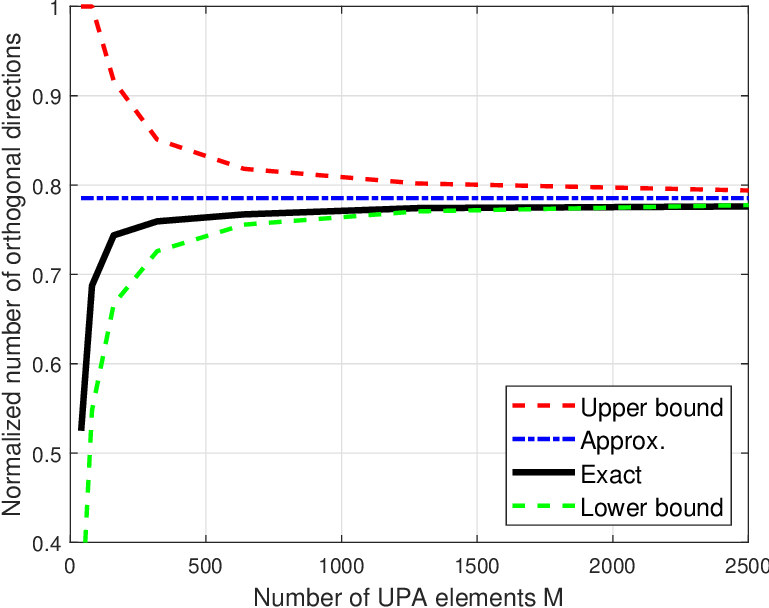}
\caption{Normalized number of orthogonal directions versus number of UPA elements $M$.\vspace{-5ex}}\label{F:numdirectionsUPA}
\end{figure}

\section{UAV Swarm Formation Optimization}\label{sec:optimization}
In the previous section, we have characterized the capacity of MU-MIMO UAV swarm communication by finding the closed-form solutions to the capacity and rate maximization problems \eqref{P:Cstar} and \eqref{P:Rstar}. This is achieved by finding the optimal solutions to the alternative problem \eqref{P:CUB} to satisfy the channel orthogonality condition \eqref{eq:orthgogonal4}. However, in many practical scenarios, the optimal solution to problem \eqref{P:CUB} does not necessarily satisfy channel orthogonality, say when the number of UAVs $K$ is greater than the number of orthogonal directions identified in \eqref{eq:NphiULA} and \eqref{eq:NExact}, or when stringent conditions like the collision avoidance constraint C4 and swarm cohesion constraint C5 are imposed. In this case, the capacity and rate maximization problems \eqref{P:Cstar} and \eqref{P:Rstar} remain unsolved. In this section, we aim to address this issue by proposing a generic algorithm to UAV swarm formation optimization problems.

Directly solving problems \eqref{P:Cstar} and \eqref{P:Rstar} is challenging, since the UAV swarm formation $\mathbf Q$ affects the capacity or achievable rate via the channel $\mathbf H$ in a sophisticated manner. To address such issues, we propose a block coordinate descent (BCD) based algorithm to optimize the UAV ranges $\{r_k\}$ and directions $\{\theta_k,\phi_k\}$ alternately. 

\subsection{Capacity Maximization with SIC}
The  3D coordinate  $\mathbf q_k$ of each UAV can be expressed in terms of its range $r_k$ and direction$(\theta_k, \phi_k)$, i.e., $\mathbf q_k = r_k \mathbf d(\theta_k, \phi_k)$, where
\begin{align}\label{eq:dk}
\mathbf d(\theta_k, \phi_k) =
\begin{bmatrix}
\cos\theta_k \cos\phi_k \\
\cos\theta_k \sin\phi_k\\
\sin\theta_k
\end{bmatrix}
\end{align}
denotes the UAV direction vector. Therefore, the capacity maximization problem \eqref{P:Cstar} can be equivalently written as
\begin{align}\label{P:CstarPolar}
\begin{cases} \underset{\{r_k\}, \{\theta_k, \phi_k\}}{\max} & C(\{r_k,\theta_k,\phi_k\})=\log_2\left|\mathbf I_M + \sum_{k=1}^K \frac{\bar P_k r_0^2}{r_k^2} \mathbf a(\theta_k, \phi_k)\mathbf a^H(\theta_k, \phi_k) \right| \\
 \text{s.t.} &\mathrm{C}_1: r_{\min} \leq r_k \leq r_{\max}, \forall k,\\
 & \mathrm{C}_2: - \Theta \leq \theta_k \leq  \Theta, \forall k,\\
 & \mathrm{C}_3: - \Phi \leq \phi_k \leq  \Phi, \forall k,\\
  & \mathrm{C}_4: \|r_k \mathbf d(\theta_k, \phi_k)-r_j \mathbf d(\theta_j, \phi_j)\|\geq d_{\min}, \forall j< k,\\
   & \mathrm{C}_5: \|r_k \mathbf d(\theta_k, \phi_k)-r_j \mathbf d(\theta_j, \phi_j)\|\leq d_{\max}, \forall j< k.\\
 \end{cases}
\end{align}

\subsubsection{UAV Swarm Range Optimization}
For any given UAV directions $\{\theta_k, \phi_k\}$ satisfying constraints $\mathrm{C}_2$ and $\mathrm{C}_3$, the sub-problem to optimize UAV ranges $\{r_k\}$ can be expressed as
\begin{align}\label{P:CstarOptimizeRange}
\begin{cases} \underset{\{r_k, b_k\}}{\max} & C(\{b_k\})=\log_2\left|\mathbf I_M + \sum_{k=1}^K b_k\mathbf a_k\mathbf a_k^H \right| \\
 \text{s.t.} & 0\leq b_k \leq \frac{\bar P_k r_0^2}{r_k^2}, \forall k, \\
&  r_{\min} \leq r_k \leq r_{\max}, \forall k,\\
   & \|r_k \mathbf d_k -r_j \mathbf d_j\|\geq d_{\min}, \forall j< k,\\
   & \|r_k \mathbf d_k -r_j \mathbf d_j\|\leq d_{\max}, \forall j< k,\\
 \end{cases}
\end{align}
where $\mathbf a_k\triangleq \mathbf a (\theta_k, \phi_k)$ and $\mathbf d_k\triangleq\mathbf d(\theta_k,\phi_k)$ are the given array steering vector and UAV direction vector, respectively.
Note that without loss of optimality, we have introduced slack variables $\{b_k\}$. At the optimal solution to \eqref{P:CstarOptimizeRange}, the constraint $b_k \leq \frac{\bar P_k r_0^2}{r_k^2}$ should be satisfied with equality, since otherwise, we may always increase $b_k$ so that the objective value is further increased.

The objective function of problem \eqref{P:CstarOptimizeRange} is concave with respect to $\{b_k\}$, and the second and fourth constraints are convex with respect to $\{r_k\}$. However, problem \eqref{P:CstarOptimizeRange} is non-convex since the first and third constraints are non-convex. Fortunately, since $\frac{\bar P_k r_0^2}{r_k^2}$ and $\|r_k \mathbf d_k -r_j \mathbf d_j\|$ are convex functions with respect to $\{r_k\}$, the problem can be efficiently addressed by successive convex approximation (SCA) technique \cite{1095}. Specifically, let $\{r_{k,l}\}$ denote the currently obtained UAV range at the $l$th iteration. Then by using the fact that any convex differentiable function is  globally lower bounded by its first-order Taylor expansion, we have
\begin{align}
&\frac{\bar P_k r_0^2}{r_k^2}\geq \frac{\bar P_k r_0^2}{r_{k,l}^2}\left(3-\frac{2r_k}{r_{k,l}}\right), \ \forall r_k, \label{eq:LB1}\\
& \|r_k \mathbf d_k - r_j \mathbf d_j\|^2 \geq g_{\mathrm{lb}}(r_k,r_j), \forall r_k, r_j,\label{eq:LB2}
\end{align}
with
\begin{align}
g_{\mathrm{lb}}(r_k,r_j)\triangleq \|r_{k,l}\mathbf d_k - r_{j,l}\mathbf d_j\|^2 +
2\left(r_{k,l}-r_{j,l}\mathbf d_k^T\mathbf d_j \right)(r_k-r_{k,l})
+2\left(r_{j,l}-r_{k,l}\mathbf d_k^T\mathbf d_j \right)(r_j-r_{j,l}).
\end{align}
Note that the global lower bounds \eqref{eq:LB1} and \eqref{eq:LB2} are tight when $r_k=r_{k,l}$, $\forall k$. By replacing the first and third constraints of problem \eqref{P:CstarOptimizeRange} with the global lower bounds \eqref{eq:LB1} and \eqref{eq:LB2}, we have the following problem
\begin{align}\label{P:CstarOptimizeRange2}
\begin{cases} \underset{\{r_k, b_k\}}{\max} & C(\{b_k\})=\log_2\left|\mathbf I_M + \sum_{k=1}^K b_k\mathbf a_k\mathbf a_k^H \right| \\
 \text{s.t.} & 0\leq b_k \leq \frac{\bar P_k r_0^2}{r_{k,l}^2}\left(3-\frac{2r_k}{r_{k,l}}\right), \forall k, \\
&  r_{\min} \leq r_k \leq r_{\max}, \forall k,\\
   & g_{\mathrm{lb}}(r_k,r_j)\geq d_{\min}^2, \forall j< k,\\
   & \|r_k \mathbf d_k -r_j \mathbf d_j\|\leq d_{\max}, \forall j< k.\\
 \end{cases}
\end{align}
Problem \eqref{P:CstarOptimizeRange2} is a convex optimization problem, which can be efficiently solved via interior-point methods or tool box like CVX \cite{227}. Furthermore, thanks to the global lower bounds \eqref{eq:LB1} and \eqref{eq:LB2}, the optimal solution to problem \eqref{P:CstarOptimizeRange2} is guaranteed to be feasible to problem \eqref{P:CstarOptimizeRange}. By iteratively solving problems \eqref{P:CstarOptimizeRange2} with updated local points $\{r_{k,l}\}$, a sequence of feasible solutions to problem \eqref{P:CstarOptimizeRange} with non-decreasing objective values can be obtained. The proposed SCA-based algorithm for solving problem \eqref{P:CstarOptimizeRange} is summarized in Algorithm~\ref{Algo:OptRange}.
\begin{algorithm}[H]
\caption{Proposed SCA algorithm for UAV range optimization.}\label{Algo:OptRange}
\begin{algorithmic}[1]
\STATE {\bf Input}: UAV swarm directions $\{\theta_k,\phi_k\}_{k=1}^K$, reference SNR $\bar P_k$ and reference range $r_0$.
\STATE {\bf Output}: UAV ranges $\{r_k^\star\}_{k=1}^K$.
\STATE {\bf Initialization}: $l=0$, and initial UAV range $\{r_{k,0}\}_{k=1}^K$ to satisfy the last three constraints of \eqref{P:CstarOptimizeRange}.
\REPEAT
\STATE Solve the convex problem \eqref{P:CstarOptimizeRange2}, and denote the optimal solution as $\{r_k^\star, b_k^\star\}_{k=1}^K$
\STATE Update $l=l+1$
\STATE Update $r_{k,l}\leftarrow r_k^\star$
\UNTIL The increase of the objective value of problem \eqref{P:CstarOptimizeRange} is negligible.
\end{algorithmic}
\end{algorithm}

\subsubsection{UAV Swarm Direction Optimization}\label{sec:direcOptimization}
Next, we consider the other subproblem of \eqref{P:CstarPolar} to optimize the UAV swarm directions $\{\theta_k, \phi_k\}$ with fixed ranges $\{r_k\}$. The  problem can be written as
\begin{align}\label{P:CstarOptimizedDirection}
\begin{cases} \underset{\{\theta_k, \phi_k\}}{\max} & C(\{\theta_k,\phi_k\})=\log_2\left|\mathbf I_M + \sum_{k=1}^K b_k \mathbf a(\theta_k,\phi_k)\mathbf a(\theta_k,\phi_k)^H_k \right| \\
 \text{s.t.}  & \mathrm{C}_2: - \Theta \leq \theta_k \leq  \Theta, \forall k,\\
 & \mathrm{C}_3: - \Phi \leq \phi_k \leq  \Phi, \forall k,\\
  & \mathrm{C}_4: \|r_k \mathbf d(\theta_k,\phi_k)-r_j \mathbf d(\theta_j,\phi_j)\|\geq d_{\min}, \forall j< k,\\
   & \mathrm{C}_5: \|r_k \mathbf d(\theta_k,\phi_k)-r_j \mathbf d(\theta_j,\phi_j)\|\leq d_{\max}, \forall j< k,\\
 \end{cases}
\end{align}
where $b_k\triangleq \bar P_kr_0^2/r_k^2$ denotes the received SNR for UAV $k$ for the given range $r_k$. Solving problem \eqref{P:CstarOptimizedDirection} with convex or non-convex optimization techniques is challenging, since the UAV swarm directions $\{\theta_k, \phi_k\}$ affect the communication capacity $C(\{\theta_k,\phi_k\})$ and the constraints $\mathrm{C_4}$ and $\mathrm{C_5}$ via the array response vectors $\mathbf a(\theta_k, \phi_k)$ and direction vectors $\mathbf d(\theta_k,\phi_k)$ in rather complicated manners. To address such issues, in the following, we propose an efficient greedy algorithm, by successively adding one UAV direction $(\theta_k, \phi_k)$ so that the capacity in problem \eqref{P:CstarOptimizedDirection} is maximized.

Specifically, at the $k$th iteration, the directions $\{\theta_j,\phi_j\}_{j=1}^{k-1}$ for the previous $k-1$ UAVs have already been determined, and we wish to find the best direction $(\theta_k, \phi_k)$ for UAV $k$ so that $C_k$ is maximized, where
\begin{align}
C_k & \triangleq \log_2 \left| \mathbf I + \sum_{j=1}^{k-1}  b_j \mathbf a_j \mathbf a_j^H + b_k \mathbf a_k \mathbf a_k^H \right| \\
& = \log_2 \left| \mathbf I+ \mathbf S_k + b_k \mathbf a_k \mathbf a_k^H \right|\\
& = \log_2 \left| \mathbf I+ \mathbf S_k\right| + \log_2\left|\mathbf I + (\mathbf I+ \mathbf S_k)^{-1}b_k \mathbf a_k \mathbf a_k^H \right|\\
& = \log_2 \left| \mathbf I+ \mathbf S_k\right| + \log_2\left(1 + b_k \mathbf a_k^H(\mathbf I+ \mathbf S_k)^{-1}\mathbf a_k  \right),
\end{align}
where for notational convenience, we have defined $\mathbf a_k\triangleq\mathbf a(\theta_k,\phi_k)$ and $\mathbf S_k\triangleq \sum_{j=1}^{k-1}  b_j \mathbf a_j \mathbf a_j^H$. At the $k$th iteration, $\mathbf S_k$ is fixed, and we only need to solve the following problem
\begin{align}\label{eq:hk_greedy4}
\begin{cases} \underset{\theta_k, \phi_k}{\max} & \mathbf a^H(\theta_k,\phi_k)(\mathbf I+ \mathbf S_k)^{-1}\mathbf a(\theta_k,\phi_k) \\
 \text{s.t.} &
 \mathrm{C}_2, \ \mathrm{C}_3\\
  &  L_{kj}\leq \mathbf d^T(\theta_k, \phi_k)\mathbf d_j \leq U_{kj}, \forall j< k,\\
 \end{cases}
\end{align}
where $L_{kj}\triangleq \frac{r_k^2+r_j^2-d_{\max}^2}{2r_kr_j}$, $U_{kj}\triangleq \frac{r_k^2+r_j^2-d_{\min}^2}{2r_kr_j}$ and $\mathbf d_j\triangleq \mathbf d(\theta_j,\phi_j)$. Note that the last constraint of \eqref{eq:hk_greedy4} directly follows from $\mathrm{C_4}$ and $\mathrm{C_5}$.

By exploiting the array manifold structures of the array response vector $\mathbf a(\theta_k,\phi_k)$ in \eqref{eq:aUPA} and direction vector $\mathbf d(\theta_k,\phi_k)$ in \eqref{eq:dk}, the above problem can be efficiently solved with codebook based method. Specifically, let $\mathrm{\Theta}_{\dic}=\left\{(\theta^{(m)},\phi^{(m)})|-\Theta\leq \theta^{(m)}\leq \Theta, -\Phi\leq \phi^{(m)} \leq \Phi\right\}$ be the set of discretized angles. Then the  angle set that is feasible to problem \eqref{eq:hk_greedy4}  is
 \begin{align}\label{eq:Thetak}
 \mathrm{\Theta}_k\triangleq \left\{(\theta,\phi)\in \mathrm{\Theta}_{\dic}|L_{kj}\leq \mathbf d^T(\theta_k, \phi_k)\mathbf d_j \leq U_{kj}, \forall j< k\right\}.
 \end{align}
 Therefore,  problem \eqref{eq:hk_greedy4} can be written as
 \begin{align}\label{eq:hk_greedy5}
\underset{(\theta_k, \phi_k)\in \mathrm{\Theta}_k}{\max} & \mathbf a^H(\theta_k,\phi_k)\mathbf J_k\mathbf a(\theta_k,\phi_k),
\end{align}
where we have defined $\mathbf J_k\triangleq (\mathbf I+ \mathbf S_k)^{-1}$. Solving the above optimization problem involves matrix inversion in order to find $\mathbf J_k$ at each iteration, which incurs high complexity. The complexity can be reduced by noting the following recursive relationship
\begin{align}
\mathbf J_k =  (\mathbf I+ \mathbf S_k)^{-1}&=\left(\mathbf I+ \sum_{j=1}^{k-1}  b_j \mathbf a_j \mathbf a_j^H\right)^{-1}\\
&=\left(\mathbf I+ \mathbf S_{k-1} + b_{k-1} \mathbf a_{k-1} \mathbf a_{k-1}^H\right)^{-1}\\
&=(\mathbf I+\mathbf S_{k-1})^{-1}-\frac{(\mathbf I+\mathbf S_{k-1})^{-1}\mathbf a_{k-1}\mathbf a_{k-1}^H(\mathbf I+\mathbf S_{k-1})^{-1}}{\frac{1}{b_{k-1}}+\mathbf a_{k-1}^H(\mathbf I+\mathbf S_{k-1})^{-1}\mathbf a_{k-1}}.
\end{align}
Thus, we have the following recursive relation:
\begin{align}\label{eq:Jk}
\mathbf J_1 =& \mathbf I_M, \\
\mathbf J_k =& \mathbf J_{k-1}-\frac{b_{k-1} \mathbf J_{k-1} \mathbf a_{k-1}\mathbf a_{k-1}^H\mathbf J_{k-1}}{1+b_{k-1}\mathbf a_{k-1}^H\mathbf J_{k-1}\mathbf a_{k-1}}, k=2,\cdots, K.\label{eq:JkRecursive}
\end{align}

Based on the above discussions, the pseudo-codes for solving the swarm direction optimization problem \eqref{P:CstarOptimizedDirection} are summarized in Algorithm~\ref{Algo:OptDirection}.
\begin{algorithm}[H]
\caption{Proposed algorithm for swarm direction optimization problem \eqref{P:CstarOptimizedDirection} .}\label{Algo:OptDirection}
\begin{algorithmic}[1]
\STATE {\bf Input}: Number of UAVs $K$, UAV ranges $\{r_k\}_{k=1}^K$, SNRs $\{b_k\}_{k=1}^K$. $d_{\min}$, $d_{\max}$, $\Theta$, and $\Phi$.
\STATE {\bf Output}: UAV directions $\{\theta_k^\star, \phi_k^\star\}_{k=1}^K$
\STATE {\bf Initialization}: Construct the dictionary $\mathrm{\Theta}_{\dic}$. Let $\mathbf J_1\leftarrow \mathbf I_M$, and $k=1$.
\WHILE {$k\leq K$}
\STATE Construct $\mathrm{\Theta}_k$ based on \eqref{eq:Thetak}.
\STATE Let $(\theta_k^\star, \phi_k^\star)\leftarrow \mathrm{arg}\underset{(\theta_k, \phi_k)\in \mathrm{\Theta}_k}{\max}  \mathbf a^H(\theta_k,\phi_k)\mathbf J_k\mathbf a(\theta_k,\phi_k)$
\STATE Let $\mathbf a_k \leftarrow \mathbf a(\theta_k^\star, \phi_k^\star)$
\STATE Update $k=k+1$.
\STATE Update $\mathbf J_k = \mathbf J_{k-1}-\frac{b_{k-1} \mathbf J_{k-1} \mathbf a_{k-1}\mathbf a_{k-1}^H\mathbf J_{k-1}}{1+b_{k-1}\mathbf a_{k-1}^H\mathbf J_{k-1}\mathbf a_{k-1}}$
\ENDWHILE
\end{algorithmic}
\end{algorithm}

By combining Algorithms~\ref{Algo:OptRange} and \ref{Algo:OptDirection}, the proposed BCD algorithm for the UAV swarm formation optimization problem \eqref{P:CstarPolar} is summarized in Algorithm~\ref{Algo:OptSwarm}.
\begin{algorithm}[H]
\caption{Proposed algorithm for capacity maximization problem \eqref{P:CstarPolar}.}\label{Algo:OptSwarm}
\begin{algorithmic}[1]
\STATE {\bf Input}: Number of UAVs $K$, reference SNR $\bar P_k$, reference range $r_0$, $d_{\min}$, $d_{\max}$, $r_{\min}$, $r_{\max}$, $\Theta$, and $\Phi$.
\STATE {\bf Output}: UAV swarm formation $\{r_k^\star, \theta_k^\star, \phi_k^\star\}_{k=1}^K$.
\STATE {\bf Initialization}: UAV range $r_k\in [r_{\min}, r_{\max}]$.
\REPEAT
\STATE With the given UAV range $\{r_k\}$, optimize the direction with Algorithm~\ref{Algo:OptDirection}.
\STATE With the given UAV direction, optimize the range with Algorithm~\ref{Algo:OptRange}.
\UNTIL The increase of the objective value of problem \eqref{P:CstarPolar} is negligible.
\end{algorithmic}
\end{algorithm}
At each iteration of Algorithm~\ref{Algo:OptSwarm}, the objective value of problem \eqref{P:CstarPolar} is guaranteed to be non-decreasing. Therefore, Algorithm~\ref{Algo:OptSwarm} is guaranteed to converge.

\subsection{Rate Maximization with TIN}\label{sec:rateMax}
In this section, we consider the rate maximization problem  \eqref{P:Rstar}  by taking into account the general constraints $\mathrm{C_1}-\mathrm{C_5}$. Similar to the capacity maximization problem \eqref{P:CstarPolar}, by expressing the UAV location $\mathbf q_k$ in terms of its range $r_k$ and direction vector $\mathbf d(\theta_k, \phi_k)$, problem  \eqref{P:Rstar} can be written as
\begin{align}\label{P:RstarPolar}
\begin{cases} \underset{\{r_k\}, \{\theta_k, \phi_k\},  \{\mathbf w_k\}}{\max} & R=\sum_{k=1}^K \log_2\left(1+\frac{\frac{\bar{P}_kr_0^2}{r_k^2}|\mathbf w_k^H \mathbf a(\theta_k, \phi_k)|^2}{\sum_{j\neq k}\frac{\bar P_jr_0^2}{r_j^2}|\mathbf w_k^H \mathbf a(\theta_j, \phi_j)|^2+ \mathbf w_k^H \mathbf w_k} \right) \\
\text{s.t.} &\mathrm{C}_1-\mathrm{C}_5.
 \end{cases}
\end{align}
where we have expressed the channel vector $\mathbf h_k(\mathbf q_k)$ in terms of $r_k$ and $(\theta_k,\phi_k)$ as defined in \eqref{eq:hk}.
Problem \eqref{P:RstarPolar} involves the joint optimization of three sets of variables: the UAV ranges $\{r_k\}$, directions $\{\theta_k, \phi_k\}$ and receive beamforming vectors $\{\mathbf w_k\}$. In the following, we propose a BCD-based algorithm to optimize $\{r_k\}$ and $\{\theta_k, \phi_k, \mathbf w_k\}$ alternately.

\subsubsection{UAV Range Optimization}
For any given UAV directions $\{\theta_k, \phi_k\}$ and receive beamforming vectors $\{\mathbf w_k\}$, the sub-problem of \eqref{P:RstarPolar} to optimize the UAV range can be formulated as
\begin{align}\label{P:RstarRange}
\begin{cases} \underset{\{r_k\}}{\max} & R=\sum_{k=1}^K \log_2\left(1+\frac{\frac{c_{kk}}{r_k^2}}{1+{\sum_{j\neq k}\frac{c_{kj}}{r_j^2}}} \right) \\
 \text{s.t.} & r_{\min} \leq r_k \leq r_{\max}, \forall k,\\
  & r_k^2+r_j^2-2r_kr_j\mathbf d_k^T\mathbf d_j\geq d_{\min}^2, \forall j< k,\\
  &  r_k^2+r_j^2-2r_kr_j\mathbf d_k^T\mathbf d_j\leq d_{\max}^2, \forall j< k,\\
 \end{cases}
\end{align}
where $c_{kj}\triangleq \frac{\bar P_j r_0^2|\mathbf w_k^H \mathbf a(\theta_j, \phi_j)|^2}{\|\mathbf w_k\|^2}$. Note that the last two constraints follow from $\mathrm{C}_4$ and $\mathrm{C}_5$ in \eqref{P:CstarPolar}. By introducing variables $x_k\triangleq \frac{1}{r_k^2}$, $\forall k$, problem \eqref{P:RstarRange} can be equivalently written as
\begin{align}\label{P:RstarRange2}
\begin{cases} \underset{\{x_k\}}{\max} & R=\sum_{k=1}^K \left[\log_2\left(1+\sum_{j=1}^K{c_{kj}x_j}\right)-\log_2\left(1+{\sum_{j\neq k}c_{kj}x_j}\right)\right] \\
 \text{s.t.} & \frac{1}{r_{\max}^2} \leq x_k \leq \frac{1}{r_{\min}^2}, \forall k, \\
 & \frac{1}{x_k}+\frac{1}{x_j}-\frac{2\mathbf d_k^T\mathbf d_j}{\sqrt{x_kx_j}}\geq d_{\min}^2, \forall j< k,\\
  &  \frac{1}{x_k}+\frac{1}{x_j}-\frac{2\mathbf d_k^T\mathbf d_j}{\sqrt{x_kx_j}}\leq d_{\max}^2, \forall j< k,\\
  & x_k\geq 0, \forall k
 \end{cases}
 \end{align}
Problem \eqref{P:RstarRange2} is non-convex, but it can be efficiently solved by SCA technique. Specifically, let $x_{k,l}$ denote the local point at the $l$th iteration. Then by using the fact that any convex (concave) differentiable function is globally lower (upper) bounded by its first-order Taylor approximation, we have the following bounds
\begin{align}
 &\log_2\left(1+{\sum_{j\neq k}c_{kj}x_j}\right)\leq g_{\mathrm{ub}}\left(\{x_j\}\right),\label{eq:bound1}\\
 & \frac{1}{\sqrt{x_kx_j}} \geq g_{\mathrm{lb}}\left(x_k,x_j\right), \label{eq:bound3}\\
 &\frac{1}{x_k} \geq h_{\mathrm{lb}}(x_k) \label{eq:bound2}
\end{align}
where
\begin{align}
& g_{\mathrm{ub}}\left(\{x_j\}\right)\triangleq \log_2\left(1+{\sum_{j\neq k}c_{kj}x_{j,l}}\right)+\frac{\log_2 e}{1+\sum_{i\neq k}c_{ki}x_{i,l}}\sum_{j\neq k} c_{kj}(x_j-x_{j,l})\\
& g_{\mathrm{lb}}\left(x_k,x_j\right)\triangleq \frac{1}{\sqrt{x_{k,l}x_{j,l}}}-\frac{1}{2}x_{k,l}^{-3/2}x_{j,l}^{-1/2}\left(x_k-x_{k,l}\right)
-\frac{1}{2}x_{j,l}^{-3/2}x_{k,l}^{-1/2}\left(x_j-x_{j,l}\right)\\
&h_{\mathrm{lb}}(x_k)\triangleq \frac{1}{x_{k,l}}-\frac{1}{x_{k,l}^2}(x_k-x_{k,l})
\end{align}
By replacing those non-convex or non-concave terms with the corresponding global bounds, problem \eqref{P:RstarRange2} can be transformed into convex optimization problems. Note that depending on whether $\mathbf d_k^T\mathbf d_j$ is positive or negative, we need to apply the global lower-bound \eqref{eq:bound3} for $\frac{1}{\sqrt{x_kx_j}}$  either for the third or fourth constraint. As an illustration, assuming $\mathbf d_k^T\mathbf d_j>0$, then we have
\begin{align}\label{P:RstarRange3}
\begin{cases} \underset{\{x_k\}}{\max} & R=\sum_{k=1}^K \left[\log_2\left(1+\sum_{j=1}^K{c_{kj}x_j}\right)-g_{\mathrm{ub}}\left(\{x_j\}\right)\right] \\
 \text{s.t.} & \frac{1}{r_{\max}^2} \leq x_k \leq \frac{1}{r_{\min}^2}, \forall k, \\
 & h_{\mathrm{lb}}(x_k)+h_{\mathrm{lb}}(x_j)-\frac{2\mathbf d_k^T\mathbf d_j}{\sqrt{x_kx_j}}\geq d_{\min}^2, \forall j< k,\\
  &  \frac{1}{x_k}+\frac{1}{x_j}-2\mathbf d_k^T\mathbf d_j g_{\mathrm{lb}}\left(x_k,x_j\right)\leq d_{\max}^2, \forall j< k,\\
  & x_k\geq 0, \forall k
 \end{cases}
 \end{align}
 Problem \eqref{P:RstarRange3} is a convex optimization problem, which can be efficiently solved via tool box like CVX. Besides, thanks to the global bounds \eqref{eq:bound1}--\eqref{eq:bound2}, the optimal solution to problem \eqref{P:RstarRange3} is guaranteed to be feasible to the non-convex optimization problem \eqref{P:RstarRange2}, and it provides a lower bound to to the optimal value of \eqref{P:RstarRange2}. By iteratively updating the local points $\{x_{k,l}\}$ and solving a sequence of the corresponding convex optimization problems \eqref{P:RstarRange3}, problem \eqref{P:RstarRange2} can be efficiently optimized.

\subsubsection{UAV Direction and Beamforming Optimization}
Next, we consider the other sub-problem of \eqref{P:RstarPolar} to jointly optimize the UAV direction $\{\theta_k,\phi_k\}$ and receive beamforming vectors $\{\mathbf w_k\}$, for any fixed UAV ranges $\{r_k\}$. The problem can be written as
\begin{align}\label{P:RstarDirection}
\begin{cases} \underset{\{\theta_k, \phi_k\},  \{\mathbf w_k\}}{\max} & R=\sum_{k=1}^K \log_2\left(1+\frac{b_k|\mathbf w_k^H \mathbf a(\theta_k, \phi_k)|^2}{\sum_{j\neq k}b_j|\mathbf w_k^H \mathbf a(\theta_j, \phi_j)|^2+ \mathbf w_k^H \mathbf w_k} \right) \\
 \text{s.t.}  & \mathrm{C}_2-\mathrm{C}_5,
 \end{cases}
\end{align}
where $b_k\triangleq \frac{\bar{P}_kr_0^2}{r_k^2}$. For any given $\{\theta_k,\phi_k\}$, the optimal receive beamforming is the LMMSE solution given by
\begin{align}
\mathbf w_k^\star = \sqrt{b_k}\left(\mathbf I+ \sum_{j\neq k} b_j \mathbf a(\theta_j, \phi_j)\mathbf a^H(\theta_j, \phi_j)\right)^{-1}\mathbf a(\theta_k, \phi_k),
\end{align}
and the resulting SINRs can be expressed as a function of the UAV directions $\{\theta_k,\phi_k\}$
\begin{align}
\SINR_k = b_k \mathbf a_k^H \Big(\sum_{j\neq k} b_j \mathbf a_j\mathbf a_j^H + \mathbf I \Big)^{-1} \mathbf a_k,\label{eq:SINRk2}
\end{align}
where for convenience, we have used the notation $\mathbf a_k\triangleq \mathbf a(\theta_k, \phi_k), \forall k$.
Therefore, problem \eqref{P:RstarDirection} is  reduced to
\begin{align}\label{P:RstarDirection2}
\begin{cases} \underset{\{\theta_k, \phi_k\}}{\max} & R=\sum_{k=1}^K \log_2\left(1+b_k \mathbf a_k^H \Big(\sum_{j\neq k} b_j \mathbf a_j\mathbf a_j^H + \mathbf I \Big)^{-1} \mathbf a_k \right) \\
 \text{s.t.}  & \mathrm{C}_2-\mathrm{C}_5.
 \end{cases}
\end{align}
Problem \eqref{P:RstarDirection2} is similar to the capacity maximization problem \eqref{P:CstarOptimizedDirection}, except for the difference on the objective functions. Therefore, the similar sequential UAV direction optimization technique can be applied.

Consider the $k$th iteration, where the directions $\{\theta_j,\phi_j\}_{j=1}^{k-1}$ of the previous $k-1$ UAVs have already been determined. Now we need to determine the  best direction $(\theta_k,\phi_k)$ for UAV $k$. Different from the capacity maximization problem \eqref{P:CstarOptimizedDirection}, the case for TIN needs a different treatment since adding one additional UAV also affects the rate of the previous UAVs due to the interference.

Define $\mathbf S_k \triangleq \sum_{j=1}^{k-1} b_j\mathbf a_j\mathbf a_j^H$. Then on the one hand, the direction $(\theta_k,\phi_k)$ for UAV $k$ should be chosen so that its own communication  rate $\log_2\left(1+b_k\mathbf a_k^H (\mathbf I + \mathbf S_k)^{-1} \mathbf a_k\right)$ is maximized. On the other hand, its adversary impact to the communication rate for the previous $k-1$ UAVs should be minimized. For $j<k$, if UAV $k$ is placed at the direction $(\theta_k,\phi_k)$, we have

\begin{align}
\SINR_j(\mathbf a_k)&=b_j\mathbf a_j^H\left(\mathbf I + \sum_{l<k, l\neq j} b_l\mathbf a_l\mathbf a_l^H +b_k \mathbf a_k \mathbf a_k^H  \right)^{-1}\mathbf a_j\\
&=b_j\mathbf a_j^H\left(\mathbf I + \mathbf S_{j,k}+b_k \mathbf a_k \mathbf a_k^H \right)^{-1}\mathbf a_j\\
&= b_j\mathbf a_j^H \mathbf J_{j,k}\mathbf a_j - \frac{b_j\mathbf a_j^H  \mathbf J_{j,k}b_k \mathbf a_k \mathbf a_k^H\mathbf J_{j,k} \mathbf a_j}{1+b_k \mathbf a_k^H \mathbf J_{j,k} \mathbf a_k}\\
&= \SINR_{j,k-1} - \frac{b_jb_k\left|\mathbf a_j^H  \mathbf J_{j,k}\mathbf a_k\right|^2}{1+b_k \mathbf a_k^H \mathbf J_{j,k} \mathbf a_k},
\end{align}
where for $j<k$, we have defined $\mathbf S_{j,k}\triangleq \sum_{l<k, l\neq j} b_l\mathbf a_l\mathbf a_l^H$ and $\mathbf J_{j,k}\triangleq \left( \mathbf I + \mathbf S_{j,k}\right)^{-1}$. Furthermore, $\SINR_{j,k-1}\triangleq b_j\mathbf a_j^H \mathbf J_{j,k}\mathbf a_j$ is the SINR of UAV $j$ before UAV $k$ is added. The second term accounts for the SINR reduction for UAV $j<k$ due to the addition of UAV $k$ that has the steering vector $\mathbf a_k$. Obviously, if UAV $k$ is the only interference to UAV $j$ so that $\mathbf S_{j,k}=\boldsymbol 0$ and $\mathbf J_{j,k}=\mathbf I$, then $\mathbf a_k$ should be orthogonal to $\mathbf a_j$ so that the SINR reduction is $0$. This is not the case when $\mathbf S_{j,k}\neq \mathbf 0$, i.e., when UAV $j$ also receives interference from other UAVs since the optimal receive beamforming vectors need to balance the interference from different UAVs.

As a result, at the $k$th iteration, the direction $(\theta_k,\phi_k)$ for UAV $k$ can be found by solving the following problem
\begin{align}\label{P:RstarDirection3}
\begin{cases} \underset{\theta_k, \phi_k}{\max} & \log_2(1+\SINR_k(\mathbf a_k))+ \sum_{j=1}^{k-1} \log_2\left(1+\SINR_j(\mathbf a_k)\right)\\
 \text{s.t.}  & \mathrm{C}_2-\mathrm{C}_5.
 \end{cases}
 \end{align}
 Similar to problem \eqref{eq:hk_greedy4}, by exploiting the array manifold structure for $\mathbf a_k$, problem \eqref{P:RstarDirection3} can be efficiently solved with the codebook based method. Furthermore, to avoid evaluating the matrix inversion  $\mathbf J_{j,k}\triangleq \left( \mathbf I + \mathbf S_{j,k}\right)^{-1}$, the following recursive relation is used
\begin{align}
\mathbf J_{j,k}&=(\mathbf I + \mathbf S_{k}- b_j \mathbf a_j\mathbf a_j^H)^{-1}\\
&=(\mathbf I + \mathbf S_{k})^{-1}-\frac{(\mathbf I + \mathbf S_{k})^{-1} b_j \mathbf a_j \mathbf a_j^H(\mathbf I + \mathbf S_{k})^{-1}}{b_j \mathbf a_j^H (\mathbf I + \mathbf S_{k})^{-1}\mathbf a_j-1 }\\
&=\mathbf J_{k}-\frac{\mathbf J_{k}b_{j} \mathbf a_{j} \mathbf a_{j}^H\mathbf J_{k}}{b_{j} \mathbf a_{j}^H\mathbf J_{k}\mathbf a_{j}-1},\label{eq:Jjk}
\end{align}
where $\mathbf J_k\triangleq (\mathbf I+\mathbf S_k)^{-1}$, which can be computed recursively based on \eqref{eq:JkRecursive} without computing matrix inversion.

Therefore, the pseudo-codes for solving the swarm direction optimization problem \eqref{P:RstarDirection2} are summarized in Algorithm~\ref{Algo:RateOptDirection}.
\begin{algorithm}[H]
\caption{Proposed algorithm for swarm direction optimization problem \eqref{P:RstarDirection2} for rate maximization.}\label{Algo:RateOptDirection}
\begin{algorithmic}[1]
\STATE {\bf Input}: Number of UAVs $K$, SNRs $\{b_k\}_{k=1}^K$. $d_{\min}$, $d_{\max}$, $\Theta$, and $\Phi$.
\STATE {\bf Output}: UAV directions $\{\theta_k^\star, \phi_k^\star\}_{k=1}^K$
\STATE {\bf Initialization}: Construct the dictionary $\mathrm{\Theta}_{\dic}$. Let $\mathbf J_1\leftarrow \mathbf I_M$, and $k=1$.
\WHILE {$k\leq K$}
\STATE Construct $\mathrm{\Theta}_k$ based on \eqref{eq:Thetak}.
\STATE For $j<k$, calculate $\mathbf J_{j,k}$ based on \eqref{eq:Jjk}
\STATE Let $(\theta_k^\star, \phi_k^\star)\leftarrow \mathrm{arg}\underset{(\theta_k, \phi_k)\in \mathrm{\Theta}_k}{\max} \log_2(1+\SINR_k(\mathbf a_k))+ \sum_{j=1}^{k-1} \log_2\left(1+\SINR_j(\mathbf a_k)\right)$
\STATE Let $\mathbf a_k \leftarrow \mathbf a(\theta_k^\star, \phi_k^\star)$
\STATE Update $k=k+1$.
\STATE Update $\mathbf J_k = \mathbf J_{k-1}-\frac{b_{k-1} \mathbf J_{k-1} \mathbf a_{k-1}\mathbf a_{k-1}^H\mathbf J_{k-1}}{1+b_{k-1}\mathbf a_{k-1}^H\mathbf J_{k-1}\mathbf a_{k-1}}$
\ENDWHILE
\end{algorithmic}
\end{algorithm}


By combining the SCA algorithm for the range optimization problem \eqref{P:RstarRange} and Algorithm~\ref{Algo:RateOptDirection} for the direction optimization, then the UAV swarm formation optimization problem \eqref{P:RstarPolar} is efficiently solved. The pseudo-codes are quite similar to Algorithm~\ref{Algo:OptSwarm}. Since at each iteration, the objective value of problem \eqref{P:RstarPolar} is non-decreasing, the algorithm is guaranteed to converge.

\section{Simulation Results}\label{sec:simulation}
In this section, we provide simulation results to verify the theoretical analysis developed in Section~\ref{sec:capacityCharacterization} and the proposed algorithms in Section~\ref{sec:optimization}.
Unless otherwise stated, the following three UAV swarm formation schemes are compared:

(1) {\it random swarm}: This is a benchmark scheme to reflect the reality of most existing MU-MIMO communication systems where the UE locations/formations are uncontrollable. In this case, the aerial UEs are randomly located within the range-angular interval specified in \eqref{eq:C1}-\eqref{eq:C5}. If no inter-UAV separation constraint is imposed, i.e., $d_{\min}=0$ and $d_{\max}=\infty$ for constraints \eqref{eq:C4} and \eqref{eq:C5}, it is not difficult to see that at the optimal solution to capacity and rate maximization problems \eqref{P:Cstar} and \eqref{P:Rstar}, all UAVs should have the minimum distance with the BS, i.e., $r_k = r_{\min}, \forall k$. In this case, the UAV angles are independently and uniformly selected from the specified angular range. On the other hand, if the minimum and/or maximum UAV separation constraints \eqref{eq:C4} and \eqref{eq:C5} are imposed, we use a greedy scheme similar to Section~\ref{sec:direcOptimization} to ensure that the obtained random swarm is feasible to all constraints in \eqref{eq:C1}-\eqref{eq:C5} to ensure fair comparison.

(2) {\it SIC, proposed swarm}: This corresponds to the proposed Algorithm~\ref{Algo:OptSwarm}, which optimizes the UAV swarm formation to maximize the sum capacity by assuming SIC.

(3) {\it TIN, proposed swarm}: This corresponds to the proposed Algorithm in Section~\ref{sec:rateMax}, which optimizes the UAV swarm formation to maximize the achievable sum rate by assuming TIN.

Both ULA and UPA are considered. Unless otherwise stated, for ULA, $M=16$ antennas are placed along the y axis, and two different values for the azimuth angular ranges are considered: $\Phi=90^\circ$ and $60^\circ$. For UPA, $M=8\times 8$ antennas are placed along the y-z plane, and two sets of the elevation and azimuth angular ranges are considered: $(\Theta, \Phi)=(90^\circ, 90^\circ), (60^\circ, 60^\circ)$. The reference range is set to $r_0=100$m, and SNR at the reference range is $\bar P=20$ dB. The minimum and maximum allowable distances between each UAV and the BS are $r_{\min}= 50$ m and $r_{\max} = 500$ m, respectively. Besides, we let $\rho = \bar P\frac{r_0^2}{r_{\min}^2}$.
Regarding the inter-UAV separation for collision avoidance and swarm cohesion constraints, two cases are examined. Case 1: No constraint, corresponding to $d_{\min}=0$ and $d_{\max}=\infty$.
Case 2: With constraint: $d_{\min}=10$ m and $d_{\max}=500$ m.

\begin{figure}
\centering
\begin{subfigure}{0.45\textwidth}
\includegraphics[width=0.85\linewidth]{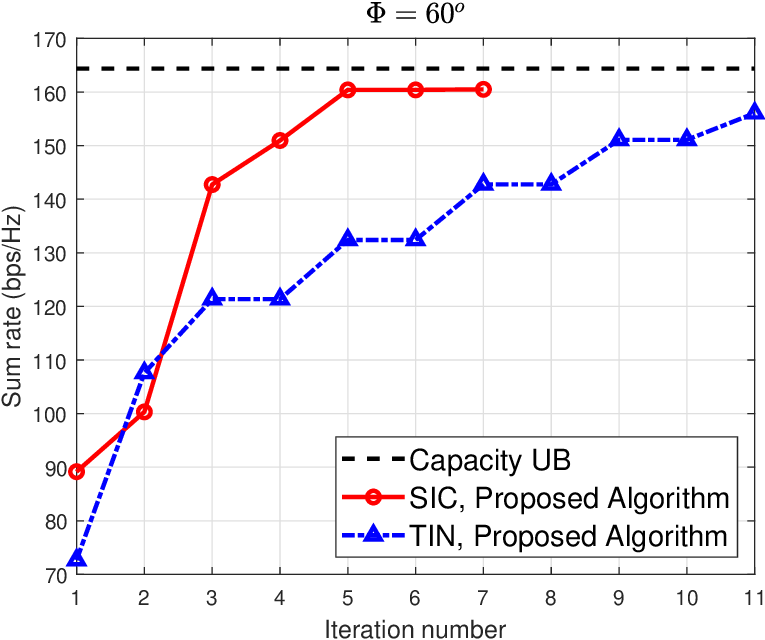}
\caption{ULA}
\end{subfigure}
\begin{subfigure}{0.45\textwidth}
\includegraphics[width=0.85\linewidth]{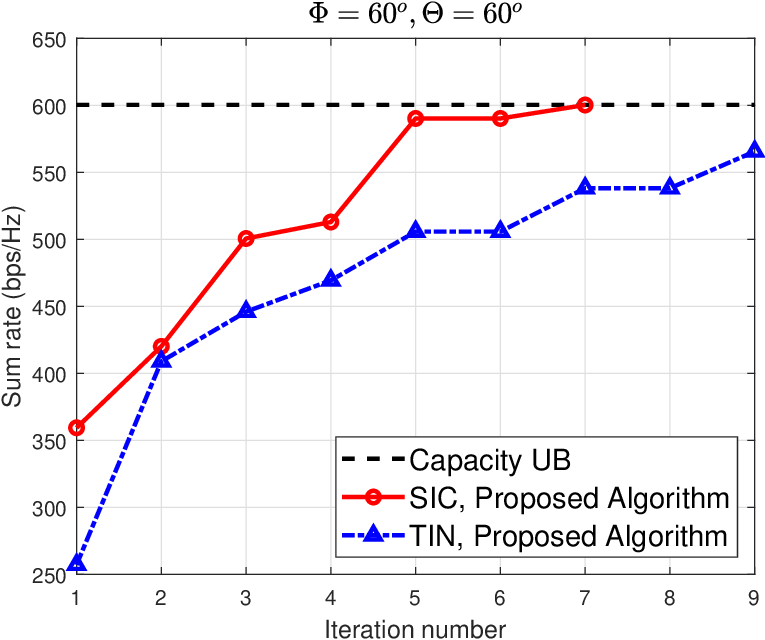}
\caption{UPA}
\end{subfigure}
\caption{Convergence of the proposed UAV swarm formation optimization algorithms.}\vspace{-2ex}\label{F:convergence}
\end{figure}

\subsection{Convergence of Proposed Algorithms}
Fig.~\ref{F:convergence} gives the convergence plot of the proposed Algorithms for formation optimization in Section~\ref{sec:optimization} for ULA and UPA, by showing the achieved sum capacity/rate of each iteration with SIC and TIN, respectively. The number of UAVs in the swarm is set to the optimal number of orthogonal directions derived in \eqref{eq:NphiULA} and \eqref{eq:NExact}, i.e., $K=13$ for ULA with $M=16$ and $\Phi = 60^\circ$, and $K=41$ for UPA with $M=64$ and $\Phi=\Theta = 60^\circ$. The minimum and maximum UAV separation thresholds are $d_{\min}=10$ m and $d_{\max}=500$ m. The curve labelled as ``Capacity UB`` corresponds to the closed-form theoretical expression given in \eqref{eq:CStarMax0} for ULA and \eqref{eq:CStarMaxUPA2} for UPA, which are the exact/approximate capacity if there is no inter-UAV separation constraint, and therefore serving as an upper bound for the considered scenario. Fig.~\ref{F:convergence} shows that for both ULA and UPA, our proposed iterative algorithms have steady increase for the achieved sum capacity/rate, and they converge with a few iterations. Furthermore, the converged capacity/rate of the proposed algorithms achieve near-optimal performance as the theoretical upper bound, which validates  theoretical analysis in Section~\ref{sec:capacityCharacterization} and demonstrates the effectiveness of the proposed optimization algorithms in Section~\ref{sec:optimization}.


\begin{figure}
\centering
\begin{subfigure}{0.45\textwidth}
\includegraphics[width=0.8\linewidth]{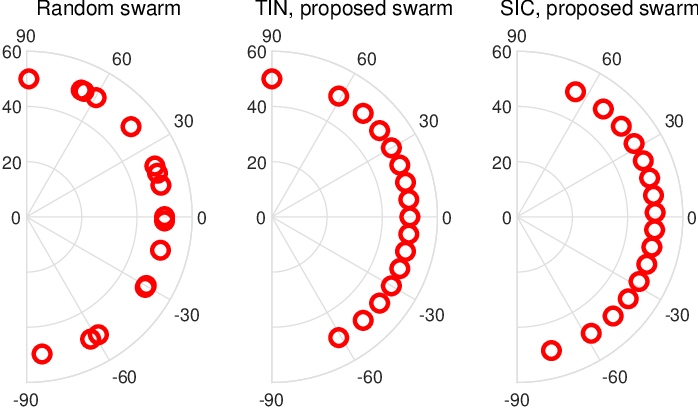}
\caption{w/o UAV separation constraint, $\Phi=90^\circ$}
\end{subfigure} \vspace{3ex}
\hspace{0.02\textwidth}
\begin{subfigure}{0.45\textwidth}
\includegraphics[width=0.8\linewidth]{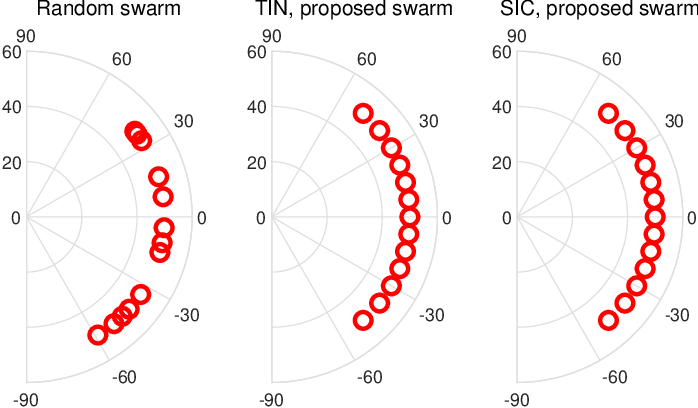}
\caption{w/o UAV separation constraint, $\Phi=60^\circ$}
\end{subfigure}\vspace{3ex}
\begin{subfigure}{0.45\textwidth}
\includegraphics[width=0.8\linewidth]{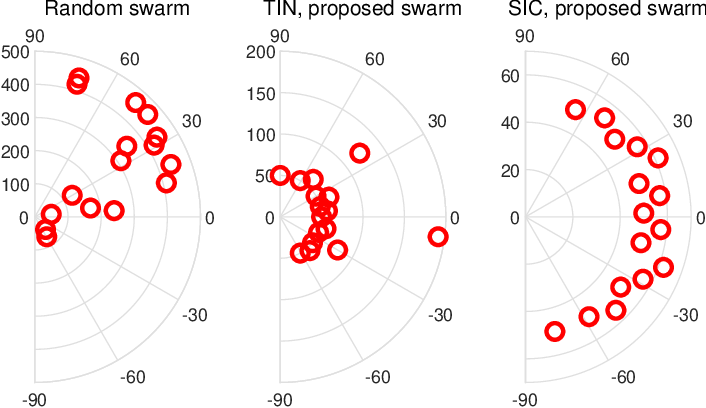}
\caption{with UAV separation constraint, $\Phi=90^\circ$}
\end{subfigure}
\begin{subfigure}{0.45\textwidth}
\includegraphics[width=0.8\linewidth]{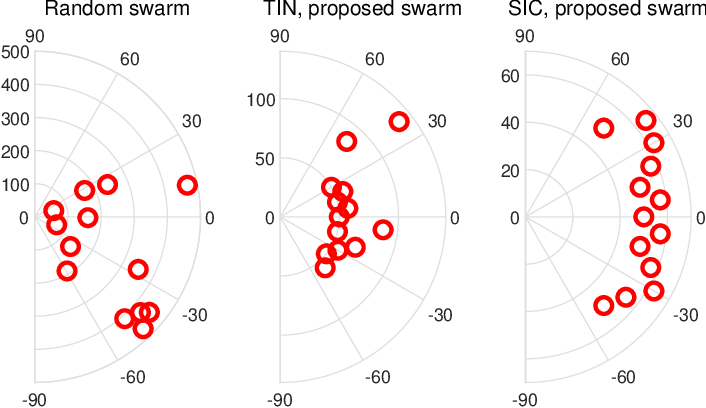}
\caption{with UAV separation constraint, $\Phi=60^\circ$}
\end{subfigure}
\caption{Polar plot of UAV swarm locations for three schemes: random swarm; proposed swarm for TIN; proposed swarm for SIC.}\label{F:swarmlocations}
\end{figure}

\subsection{UAV Swarm Locations and Interference Pattern}
Next, we examine the resulting UAV swarm locations for the three schemes mentioned above: random swarm and optimized swarms with TIN and SIC respectively. For ease of visualization, the ULA case is considered and the resulting UAV swarm locations are plotted in a polar coordinate showing the corresponding UAV range and azimuth angle $(r_k, \phi_k)$. It is observed from Fig.~\ref{F:swarmlocations}(a) and (b) that if no inter-UAV separation constraint is applied, all UAVs have the minimum distance from the BS, i.e., $r_k=r_{\min}=50$ m $\forall k$, while this is not true if inter-UAV separation constraints are considered, as shown in Fig.~\ref{F:swarmlocations}(c) and (d). This is expected since for the latter, the UAV may have to move further away from the BS in order to ensure the minimum UAV separations and also achieve interference mitigation for the TIN case. In terms of UAV directions, different from the conventional random UE cases, the optimized swarm UEs with our proposed algorithms achieve effective coordination for the UE directions so that the inter-UE interference is minimized. In fact, for the case without UAV separation constraint shown in Fig.~\ref{F:swarmlocations} (a) and (b), the proposed formation optimization algorithms give the same UAV directions as the theoretical optimal values derived in Section~\ref{sec:capacityCharacterization}, which achieves zero inter-UE interference.

 For the case with inter-UAV separation constraints shown in Fig.~\ref{F:swarmlocations} (c) and (d), zero inter-UE interference may not be achievable for all UAV pairs. Notably, the proposed algorithms still achieve effective inter-UAV interference separation by UAV direction optimization. To demonstrate this, Fig.~\ref{F:interference} plots the inter-UAV interference coefficients corresponding to the UAV swarm locations in Fig.~\ref{F:swarmlocations}(d). The coefficients are given by the $K\times K$-dimensional matrix defined as $|\mathbf A^H \mathbf A|^2/M^2$, with $\mathbf A\in \mathbb{C}^{M \times K}$ being the array steering matrix with $\mathbf a(\theta_k,\phi_k)$ being its columns. It is observed from Fig.\ref{F:interference} that in contrast to the random swarm scheme that suffers from severe inter-UAV interference, the proposed algorithms achieve almost zero inter-UAV interference by exploiting the additional design degree of freedom of swarm formation optimization, without sacrificing its own beam forming gain $M$. This is fundamentally different from classic beamforming techniques like zero-forcing (ZF) beamforming, which achieves interference nulling but at the cost of sacrificing its own beamforming gain.

\begin{figure}
\centering
\begin{subfigure}{0.32\textwidth}
\includegraphics[width=0.85\linewidth]{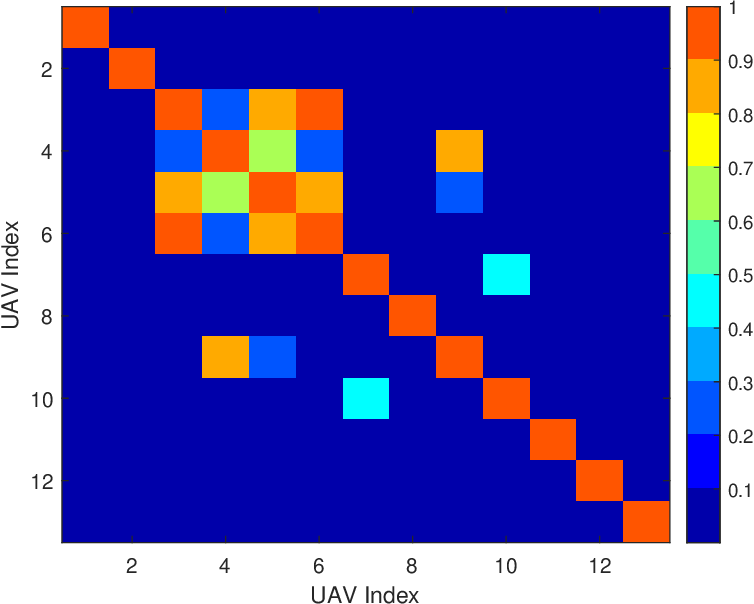}
\caption{Random swarm}
\end{subfigure}
\begin{subfigure}{0.32\textwidth}
\includegraphics[width=0.85\linewidth]{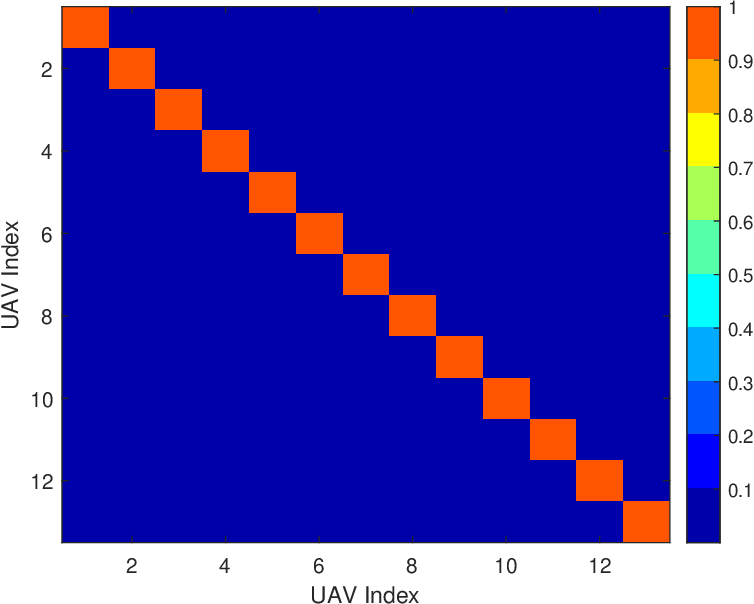}
\caption{TIN, proposed}
\end{subfigure}
\begin{subfigure}{0.32\textwidth}
\includegraphics[width=0.85\linewidth]{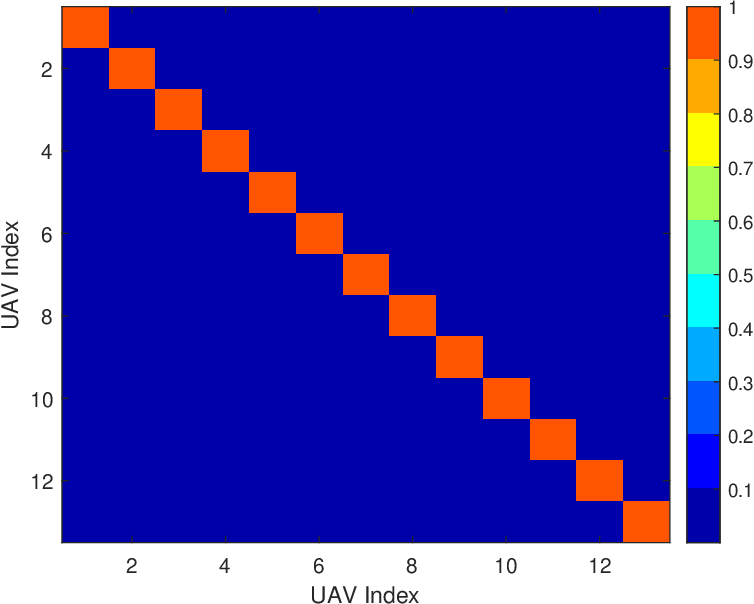}
\caption{SIC, proposed}
\end{subfigure}
\caption{Inter-UAV interference patterns for the three swarm schemes.}\label{F:interference}
\end{figure}

Fig.~\ref{F:UAVSeparations} plots the inter-UAV separations for all the $K(K-1)/2$ pairs of UAVs corresponding to the UAV swarm locations shown in Fig.~\ref{F:swarmlocations}(d). It is observed that for both random swarm and optimized swarm with our proposed algorithms, the minimum and maximum inter-UAV swarm separation constraints $d_{\min}=10$m and $d_{\max}=500$m are satisfied.

\begin{figure}
\centering
\includegraphics[width=0.5\linewidth]{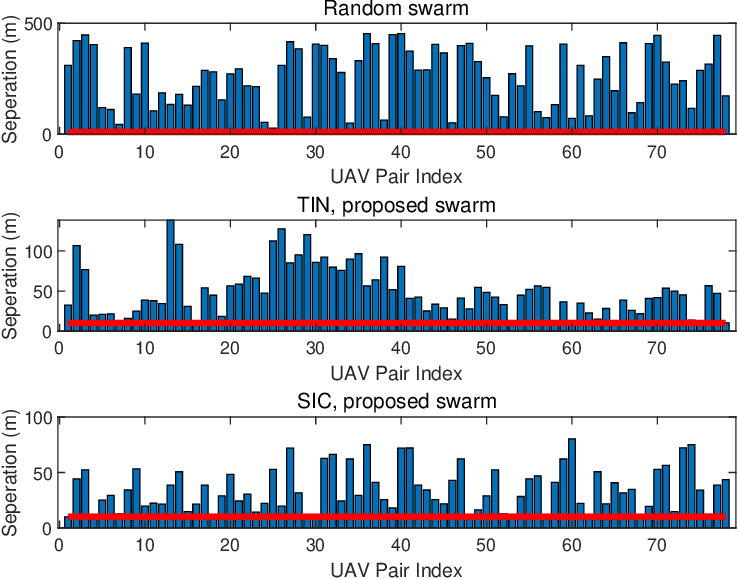}
\caption{Inter-UAV separations, with constraints $d_{\min}=10$ m and $d_{\max} = 500$ m.}\label{F:UAVSeparations}
\end{figure}

\subsection{Average Rate vs $K$}
Fig.~\ref{F:averageRate} gives the achieved average sum rate versus the number of UAV UEs $K$ for the ULA case with $M=16$. It is observed from Fig.~\ref{F:averageRate}(a) and (b) that for the case without UAV separation constraints, the proposed swarm formation optimization algorithms achieve sum rate linearly increasing with $K$ when $K\leq 16$ for $\Phi=90^\circ$ and $K\leq 13$ for $\Phi=60^\circ$, and they achieve the theoretical capacity upper bound. This is consistent with the theoretical analysis in Section~\ref{sec:optimalCapacity}, showing that if $K\leq N(\bar \Phi)$, with $N(\bar \Phi)$ given by \eqref{eq:NphiULA}, the optimal capacity is $C=K\log_2(1+\rho M)$. Furthermore, Fig.~\ref{F:averageRate}(a) and (b) demonstrate that our proposed UAV swarm formation optimization algorithms are able to achieve the theoretically optimal solution for the considered setups, and they significantly outperform the benchmarking random swarm for both SIC and TIN cases. When inter-UAV separation constraints are imposed, Fig.~\ref{F:averageRate}(c) and (d) demonstrate that though with slight performance degradation, our proposed algorithms that exploit the UEs' controllable mobility still achieve tremendous gains over the baseline schemes which fail to utilize such new design degree of freedom. For instance, for the case of TIN with $K=16$, the proposed schemes achieve a sum rate increase by about $1.6$ times and $2$ times for $\Phi=90^\circ$ and $\Phi=60^\circ$, respectively.

\begin{figure}
\centering
\begin{subfigure}{0.45\textwidth}
\includegraphics[width=0.8\linewidth]{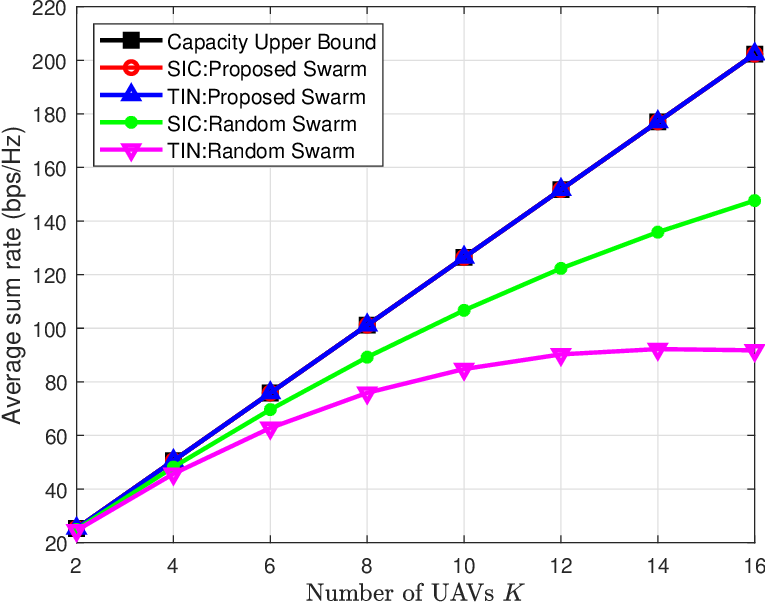}
\caption{w/o UAV separation constraint, $\Phi=90^\circ$}
\end{subfigure}\vspace{3ex}
\begin{subfigure}{0.45\textwidth}
\includegraphics[width=0.8\linewidth]{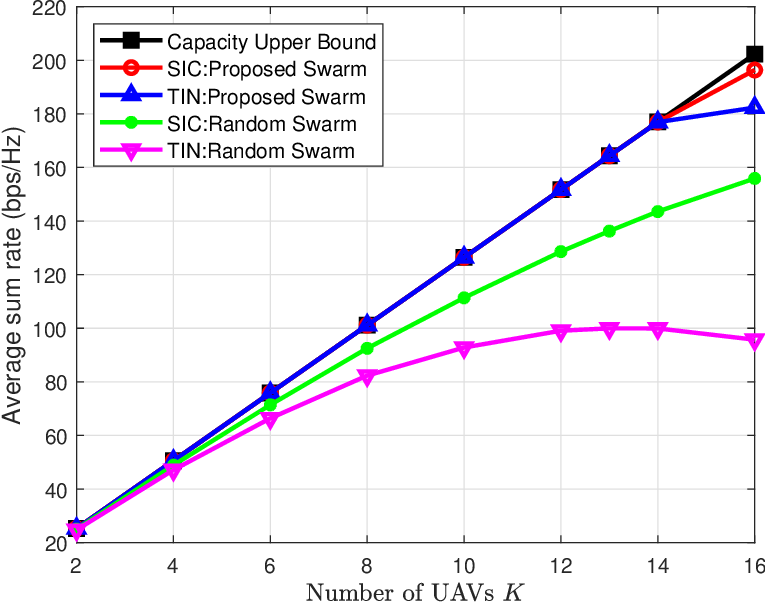}
\caption{w/o UAV separation constraint, $\Phi=60^\circ$}
\end{subfigure}\vspace{3ex}
\begin{subfigure}{0.45\textwidth}
\includegraphics[width=0.8\linewidth]{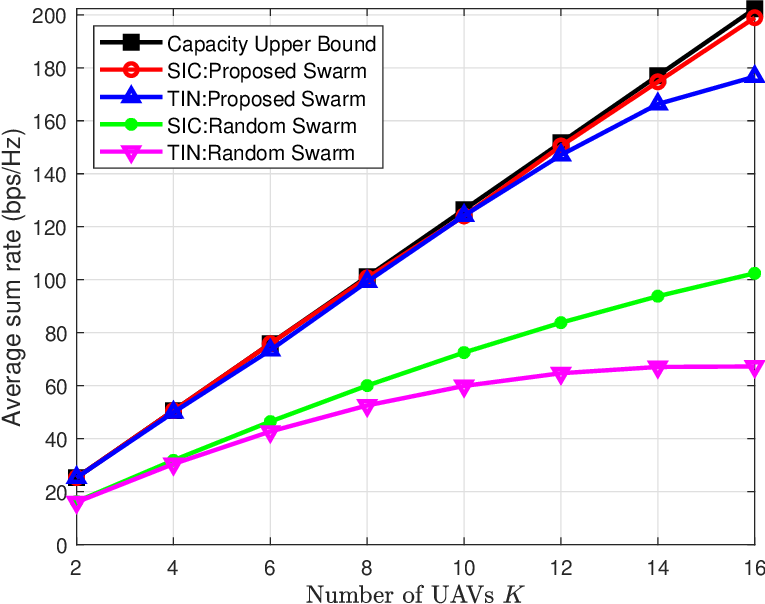}
\caption{with UAV separation constraint, $\Phi=90^\circ$}
\end{subfigure}
\begin{subfigure}{0.45\textwidth}
\includegraphics[width=0.8\linewidth]{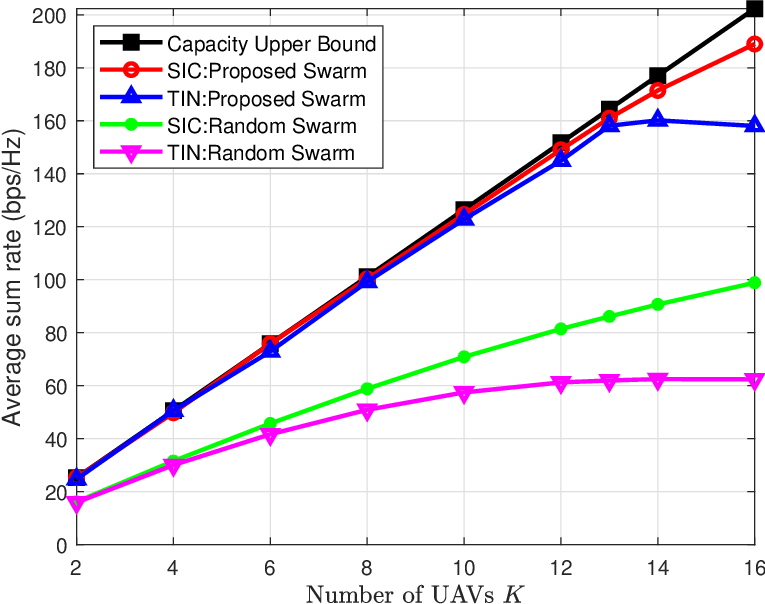}
\caption{with UAV separation constraint, $\Phi=60^\circ$}
\end{subfigure}
\caption{Average sum rate versus the number of UAVs $K$.}\label{F:averageRate}
\end{figure}

\subsection{Rate CDF}
Fig.~\ref{F:rateCDF} gives the cumulative distribution function (CDF) plot of the achieved communication sum rate for random swarm and proposed optimized swarm for both ULA and UPA, considering the case with inter-UAV separation constraints. The number of UAVs $K$ are set based on the theoretical values $N(\bar \Phi)$ in \eqref{eq:NphiULA} and  for ULA and $N_{\mathrm{UPA}}$ in \eqref{eq:NExact} for UPA. For the case of random swarm, $10000$ random realizations of the feasible UAV locations are considered. It is observed from the figure that the proposed swarm is able to achieve guaranteed deterministic performance that matches or approach to the capacity upper bound, thanks to its exploitation of the UEs' controllable mobility via swarm formation optimization. By contrast, if UEs are randomly located as in the conventional MU-MIMO communication systems, there is very little chance that the resulting performance is close to the capacity upper bound. This effectively demonstrate that fully controllable mobility of swarm UEs like UAV swarms is able to boost the communication capacity dramatically.

\begin{figure}
\centering
\begin{subfigure}{0.45\textwidth}
\includegraphics[width=0.8\linewidth]{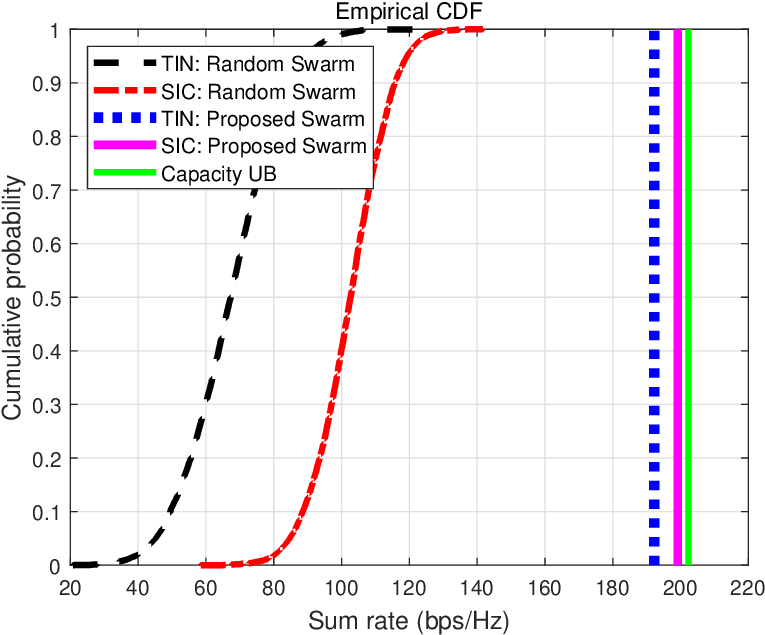}
\caption{ULA, $\Phi=90^\circ$}
\end{subfigure}\vspace{3ex}
\begin{subfigure}{0.45\textwidth}
\includegraphics[width=0.8\linewidth]{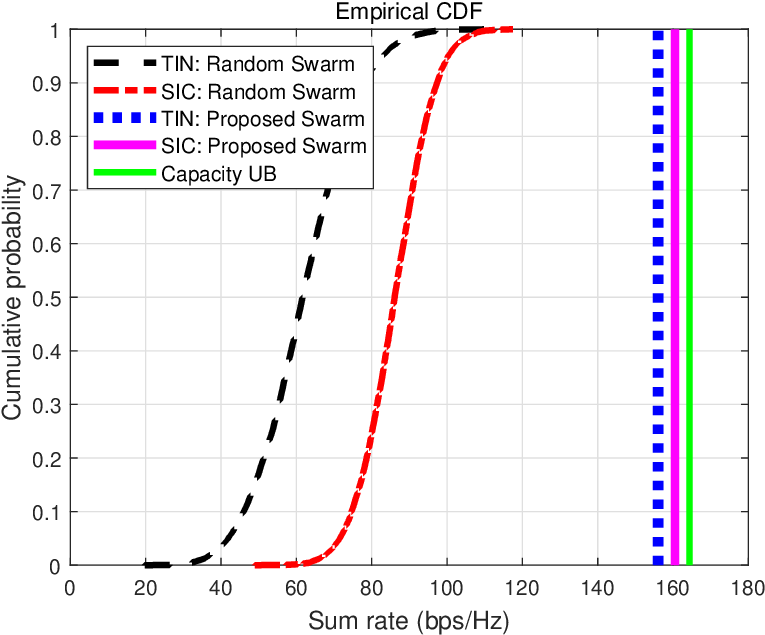}
\caption{ULA, $\Phi=60^\circ$}
\end{subfigure}\vspace{3ex}
\begin{subfigure}{0.45\textwidth}
\includegraphics[width=0.8\linewidth]{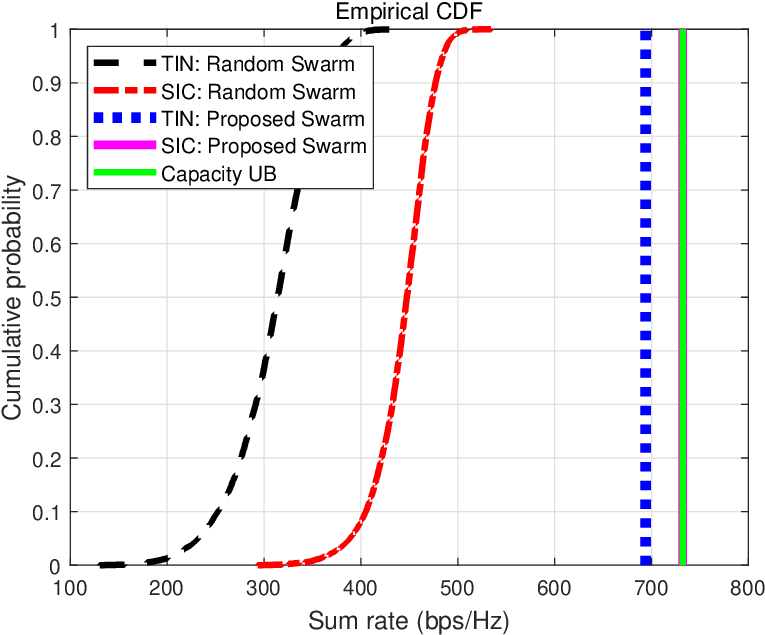}
\caption{UPA, $\Phi=\Theta=90^\circ$}
\end{subfigure}
\begin{subfigure}{0.45\textwidth}
\includegraphics[width=0.8\linewidth]{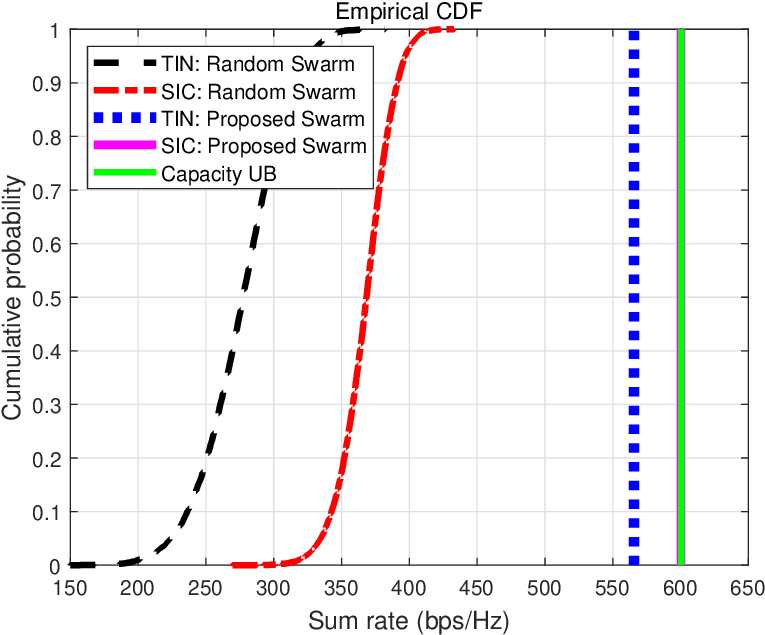}
\caption{UPA, $\Phi=\Theta=60^\circ$}
\end{subfigure}
\caption{Cumulative distribution function of communication sum rate.}\label{F:rateCDF}
\end{figure}

\section{Conclusion}\label{sec:Conclusion}
In this paper, we investigated a MU-MIMO wireless communication system in which a swarm of UAVs serves as the UEs. The key distinction from conventional MU-MIMO systems lies in the ability to optimize the UAVs' positions or formation to maximize either the sum-capacity (with SIC) or the sum-rate (with TIN). By leveraging this new degree of freedom, we first derived a closed-form characterization of the sum-capacity for MU-MIMO UAV swarm communications, revealing that full spatial multiplexing gain and beamforming gain can be simultaneously attained. For the general case incorporating practical swarm constraints such as collision avoidance and cohesion, we proposed effective algorithms to optimize the UAV formation for maximizing sum-capacity and sum-rate. Extensive numerical results were provided to validate both the theoretical analysis and the effectiveness of the proposed algorithms.

While this work extends classical MU-MIMO capacity characterization and rate optimization frameworks for UEs with controllable and coordinated mobility,  several limitations warrant further investigation. First, our current model assumes single-antenna UAVs; extending the analysis to multi-antenna UAVs presents an important direction for future work. Second, the study relies on LoS free-space channel model, whereas real-world aerial-to-ground channels in complex environments typically involve a mix of LoS and non-LoS (NLoS) components. Adapting the proposed framework to such realistic propagation conditions remains an open challenge. A promising avenue is to integrate the emerging concept of channel knowledge map (CKM) \cite{2101}, which can provide location-specific prior channel knowledge to guide swarm formation optimization in complex environments.

\appendices
\section{Proof of Theorem~\ref{theo:CUB}}\label{A:CUB}
To prove Theorem~\ref{theo:CUB}, it is first noted that since the two optimization problems share the same feasible region $\mathcal Q$, for any UAV swarm formation $\mathbf Q$ that is feasible to  problem \eqref{P:Cstar}, it is also feasible to \eqref{P:CUB}. Let $\mathbf Q_2^\star$ denote an optimal solution to problem \eqref{P:CUB}. Then $\forall \mathbf Q\in \mathcal Q$, we have the following relations:
\begin{align}
 C(\mathbf Q)&=\log_2\left|\mathbf I +  \mathbf{\bar{P}} \mathbf H^H(\mathbf Q) \mathbf H(\mathbf Q) \right| \label{eq:ineq0C}\\
 &\leq \log_2\left|\mathbf I +  \mathbf{\bar{P}} \diag\{\mathbf H^H(\mathbf Q) \mathbf H(\mathbf Q)\} \right| \label{eq:ineq1C} \\
 &=\sum_{k=1}^K \log_2\left(1+ \bar P_k \|\mathbf h_k(\mathbf Q)\|^2 \right)\\
 &\leq \sum_{k=1}^K \log_2\left(1+ \bar P_k \|\mathbf h_k(\mathbf Q_2^\star)\|^2\right)\label{eq:ineq2C} \\
 &=C_2^\star, \forall \mathbf Q\in \mathcal Q, \label{eq:ineq3C}
\end{align}
where the inequality \eqref{eq:ineq1C} holds due to the Fischer's (Hadamard's) inequality,  and the  inequality \eqref{eq:ineq2C} holds since $\mathbf Q_2^\star$ is an optimal solution to problem \eqref{P:CUB}.

 The remaining task for proving Theorem~\ref{theo:CUB} is to show that \eqref{eq:orthgogonal4} is the  sufficient and necessary conditions for the upper bound \eqref{eq:CUB} to be tight. Proving  the sufficient condition is straightforward, by simply substituting  $\mathbf Q=\mathbf Q_2^\star$ to problem \eqref{P:Cstar}, which gives $C(\mathbf Q_2^\star)=C_2^\star$ if \eqref{eq:orthgogonal4} is satisfied. Next, we show that \eqref{eq:orthgogonal4} is also the necessary condition. This is achieved by contradiction. Suppose, on the contrary, that there exists no optimal solution to problem \eqref{P:CUB} satisfying the conditions \eqref{eq:orthgogonal4}, while the upper bound \eqref{eq:CUB} is still tight, i.e., $\exists \mathbf Q' \in \mathcal Q$ such that $C(\mathbf Q')=C_2^\star$. Note that since $\mathbf Q'$ also needs to satisfy the relations \eqref{eq:ineq0C}-\eqref{eq:ineq3C}, then both inequalities \eqref{eq:ineq1C} and \eqref{eq:ineq2C} must be satisfied with equality. Furthermore, for \eqref{eq:ineq1C} to be satisfied with equality, a sufficient and necessary condition is that $\mathbf H^H(\mathbf Q') \mathbf H(\mathbf Q')$ is a diagonal matrix, which is equivalent to $\mathbf h_k^H(\mathbf Q')\mathbf h_j(\mathbf Q')=0, \forall k\neq j$. Besides, satisfying \eqref{eq:ineq2C} with equality implies that $\mathbf Q'$ is also an optimal solution to problem \eqref{P:CUB}  since it belongs to $\mathcal Q$ and also achieves the optimal objective value $C_2^\star$. This  contradicts with the contrary assumption that there exists no optimal solution to  problem \eqref{P:CUB} satisfying the orthogonality conditions \eqref{eq:orthgogonal4}.

This thus completes the proof of Theorem~\ref{theo:CUB}.

\section{Proof of Theorem~\ref{theo:1}}\label{A:theo:1}
Theorem~\ref{theo:1} can be easily shown by noticing that for any given UAV swarm formation $\mathbf Q$, the $K$ optimal receive beamforming vectors $\mathbf w_k$ to problem \eqref{P:Rstar} can be obtained independently by maximizing the $\SINR_k(\mathbf Q, \mathbf w_k)$ in \eqref{eq:SINRk}. This is a generalized Rayleigh quotient problem, and the optimal solution is given by the classic linear minimum mean square error (LMMSE) beamforming:
\begin{align}
\mathbf w_k^\star = \Big(\sum_{j\neq k} \bar P_j \mathbf h_j\mathbf h_j^H + \mathbf I_M \Big)^{-1}\sqrt{\bar P_k} \mathbf h_k.\label{eq:wkMMSE}
\end{align}
By substituting \eqref{eq:wkMMSE} into \eqref{eq:SINRk}, the resulting maximum SINR for any specified UAV swarm formation $\mathbf Q$ is given by \eqref{eq:SINRkQ}. This thus completes the proof of Theorem~\ref{theo:1}.

\section{Proof of Theorem~\ref{theo:2}}\label{A:upperbound}
In order to show Theorem~\ref{theo:2}, we need the following result:
\begin{proposition}\label{props1}
The following inequality holds
\begin{align}
\SINR_k(\mathbf Q) &\triangleq \bar P_k \mathbf h_k^H \Big(\sum_{j\neq k} \bar P_j \mathbf h_j\mathbf h_j^H + \mathbf I_M \Big)^{-1} \mathbf h_k \label{eq:SINRkQ2} \\
&\leq \bar P_k \|\mathbf h_k\|^2, ~ \forall \mathbf Q, \forall k, \label{eq:UB3}
\end{align}
where the equality holds if and only if
\begin{align}
\mathbf h_k^H \mathbf h_j = 0, \forall k\neq j. \label{eq:orthgogonal2}
\end{align}
\end{proposition}
\begin{IEEEproof}
To show Proposition~\ref{props1}, let's define $\mathbf S_k \triangleq \sum_{j\neq k} \bar P_j \mathbf h_j \mathbf h_j^H$ as the interference covariance matrix for UAV $k$. Since $\mathbf S_k$ is a summation of $K-1$ rank-1 positive semidefinite matrices, it is also positive semidefinite but possibly rank-deficient, i.e., $\mathbf S_k\succeq \mathbf 0$ and $l_k\triangleq \mathrm{rank} (\mathbf S_k)\leq M$. Let $\mathbf S_k= \mathbf U_k \boldsymbol \Lambda_k \mathbf U_k^H$ be the (reduced) eigenvalue decomposition of $\mathbf S_k$, where $\mathbf U_k \in \mathbb{C}^{M\times r_k}$ is the semi-unitary matrix, i.e., $\mathbf U_k^H \mathbf U_k = \mathbf I_{l_k}$, which is composed by the $r_k$ eigenvectors of $\mathbf S_k$ associated with those positive eigenvalues, and $\boldsymbol \Lambda_k = \diag \{\lambda_1,\cdots \lambda_{l_k}\}$ is a diagonal matrix consisting of the $l_k$ positive eigenvalues. As such, the $\SINR_k(\mathbf Q)$ in \eqref{eq:SINRkQ2} can be expressed as
\begin{align}
\SINR_k(\mathbf Q) &= \bar P_k \mathbf h_k^H \Big( \mathbf U_k \boldsymbol \Lambda_k \mathbf U_k^H+ \mathbf I_M \Big)^{-1} \mathbf h_k \\
& = \bar P_k \mathbf h_k^H \Big( \mathbf I_M - \mathbf U_k (\boldsymbol \Lambda_k^{-1} + \mathbf U_k^H \mathbf U_k)^{-1}\mathbf U_k^H\Big)  \mathbf h_k\\
& = \bar P_k \|\mathbf h_k\|^2 -\bar P_k \mathbf h_k^H \mathbf U_k (\boldsymbol \Lambda_k^{-1} + \mathbf I_{l_k})^{-1}\mathbf U_k^H \mathbf h_k \\
& \leq \bar P_k \|\mathbf h_k\|^2, ~ \forall \mathbf Q, \forall k, \label{eq:inequa}
\end{align}
where second equality follows from the matrix inversion lemma, and the last inequality follows from the fact that $(\boldsymbol \Lambda_k^{-1} + \mathbf I_{l_k})^{-1}$ is a positive definite matrix. Besides, since the quadratic form of any non-zero vector with a positive definite matrix is positive, the inequality in \eqref{eq:inequa} will become equality if and only if $\mathbf U_k^H \mathbf h_k = \mathbf 0, \forall k$.  As a result, the remaining task is to show that $\mathbf U_k^H \mathbf h_k = \mathbf 0, \forall k$ is equivalent to the conditions specified in \eqref{eq:orthgogonal2}. This can be easily shown by noticing the following equivalence relationship:
\begin{align}
\mathbf U_k^H\mathbf h_k=\boldsymbol 0 & \Longleftrightarrow \mathbf h_k^H \mathbf U_k \boldsymbol \Lambda_k \mathbf U_k^H\mathbf h_k= 0 \\
& \Longleftrightarrow \mathbf h_k^H  \big(\sum_{j\neq k} \bar P_j \mathbf h_j \mathbf h_j^H\big) \mathbf h_k= 0  \\
&  \Longleftrightarrow   \sum_{j\neq k} \bar P_j |\mathbf h_k^H \mathbf h_j|^2= 0 \\
& \Longleftrightarrow \mathbf h_k^H \mathbf h_j =0, \forall k\neq j,
\end{align}
where the first equivalence holds since $\boldsymbol \Lambda_k$ is a positive definite matrix.
This thus completes the proof of Proposition~\ref{props1}.
\end{IEEEproof}

With Proposition~\ref{props1}, we are now ready to show Theorem~\ref{theo:2}. For any UAV formation $\mathbf Q$ that is feasible to both problems  \eqref{P:CUB} and \eqref{P:RQ}, we have the following results:
\begin{align}
R(\mathbf Q) & = \sum_{k=1}^K \log_2\left( 1 + \SINR_k(\mathbf Q)\right) \label{eq:ineq0}\\
& \leq  \sum_{k=1}^K \log_2\left( 1 + \bar P_k \|\mathbf h_k(\mathbf Q)\|^2 \right) \label{eq:ineq1}\\
& \leq  \sum_{k=1}^K \log_2\left( 1 + \bar P_k \|\mathbf h_k(\mathbf Q_2^\star)\|^2 \right) \label{eq:ineq2}\\
&\triangleq C_{2}^\star, ~ \forall \mathbf Q\in \mathcal Q, \label{eq:ineq3}
\end{align}
where the first inequality follows from the global upper bound \eqref{eq:UB3} specified in Proposition~\ref{props1}, and the second equality holds since $\mathbf Q_2^\star$ is an optimal solution to problem \eqref{eq:CUB}. Thus, the upper bound specified in \eqref{eq:RUB} is shown.

 The remaining task is to show that \eqref{eq:orthgogonal4} is the  sufficient and necessary conditions for the upper bound \eqref{eq:RUB} to be tight. Proving  the sufficient condition is straightforward, by simply substituting  $\mathbf Q=\mathbf Q_2^\star$ to problem \eqref{P:RQ}, which gives $R(\mathbf Q_2^\star)=C_2^\star$ if \eqref{eq:orthgogonal4} is satisfied. Next, we show that \eqref{eq:orthgogonal4} is also the necessary condition. This is achieved by contradiction. Suppose, on the contrary, that there exists no optimal solution to problem \eqref{eq:CUB} satisfying the conditions \eqref{eq:orthgogonal4}, while the upper bound \eqref{eq:RUB} is still tight, i.e., $\exists \mathbf Q_2 \in \mathcal Q$ such that $R(\mathbf Q_2)=C_2^\star$. Note that since $\mathbf Q_2$ needs to satisfy the relations \eqref{eq:ineq0}-\eqref{eq:ineq3}, then both inequalities \eqref{eq:ineq1} and \eqref{eq:ineq2} need to be satisfied with equality. Furthermore, based on Proposition~\ref{props1}, for \eqref{eq:ineq1} to be satisfied with equality, a sufficient and necessary conditions is $\mathbf h_k^H(\mathbf Q_2)\mathbf h_j(\mathbf Q_2)=0, \forall k\neq j$. Besides, satisfying \eqref{eq:ineq2} with equality implies that $\mathbf Q_2$ is also an optimal solution to problem \eqref{eq:CUB} since it also achieves the optimal value $C_2^\star$. This obviously contradicts with the contrary assumption that there exists no optimal solution to problem \eqref{eq:CUB} satisfying the orthogonality conditions \eqref{eq:orthgogonal4}.

This thus completes the proof of Theorem~\ref{theo:2}.

\section{Proof of Theorem~\ref{theo:NUPABounds}}\label{A:NUPABounds}
By applying $\lfloor x \rfloor \leq x$ to \eqref{eq:NUPA4}, the following upper bound of $N_\UPA$ can be found:
\begin{align}
N_\UPA & \leq   M_z \bar \Theta +1 + \bar \Phi M_y  \sum_{l=- \frac{M_z\bar \Theta}{2}}^{\frac{M_z\bar \Theta}{2}}\sqrt{1-\frac{4l^2}{M_z^2}}\\
  &\approx M_z\left(\frac{ \bar \Phi  M_y \left(\Theta+\bar \Theta \cos \Theta \right)}{2} +\bar \Theta \right)  +1,
\end{align}
where the approximation is obtained by approximating summation with integration, i.e., for $L_1\leq \frac{M_z}{2}$, we have
\begin{align}
\sum_{l=-L_1}^{L_1}\sqrt{1-\frac{4l^2}{M_z^2}}&\approx \frac{M_z}{2}\int_{-\frac{2L_1}{M_z}}^{\frac{2L_1}{M_z}} \sqrt{1-x^2}dx\\
&=\frac{M_z}{2}\left[\sin^{-1}\left(\frac{2L_1}{M_z}\right) +\frac{2L_1}{M_z}\sqrt{1-\left( \frac{2L_1}{M_z}\right)^2}\right], \label{eq:approx}
\end{align}
which is a valid approximation when $M_z\gg 2$, so that $x$ in the integral is discretized by infinitesimal step size with $x=\frac{2l}{M_z}$.

Besides, for multi-antenna communication with $M=M_yM_z$ antennas, the number of orthogonal channels is known to be upper bounded by $M_yM_z$ as well. In fact, this can also be easily shown from \eqref{eq:Ntheta} and \eqref{eq:Nphi} due to the minimum operation.

Similarly, by applying $\lfloor x \rfloor \geq x-1$ to \eqref{eq:NUPA4}, we may find the lower bound for $N_\UPA$ as:
\begin{align}
N_\UPA & \geq   
1-M_z \bar \Theta + \bar \Phi M_y  \sum_{l=- (\frac{M_z\bar \Theta}{2}-1)}^{\frac{M_z\bar \Theta}{2}-1}\sqrt{1-\frac{4l^2}{M_z^2}}\\
&\approx 1-M_z\bar \Theta +  \frac{\bar \Phi M_y M_z}{2}\left[\sin^{-1}\left(\bar \Theta - \frac{2}{M_z}\right) +\left(\bar \Theta - \frac{2}{M_z}\right)\sqrt{1- \left(\bar \Theta - \frac{2}{M_z}\right)^2}\right]\\
&\approx M_z\left(\frac{ \bar \Phi  M_y \left(\Theta+\bar \Theta \cos \Theta \right)}{2} - \bar \Theta \right)  +1
\end{align}
where the first approximation follows from \eqref{eq:approx}, and the last approximation is valid when $M_z\gg \frac{2}{\bar \Theta}$.

This completes the proof of Theorem~\ref{theo:NUPABounds}.

\bibliographystyle{IEEEtran}
\bibliography{IEEEabrv,IEEEfull}

\begin{thebibliography}{10}
\providecommand{\url}[1]{#1}
\csname url@samestyle\endcsname
\providecommand{\newblock}{\relax}
\providecommand{\bibinfo}[2]{#2}
\providecommand{\BIBentrySTDinterwordspacing}{\spaceskip=0pt\relax}
\providecommand{\BIBentryALTinterwordstretchfactor}{4}
\providecommand{\BIBentryALTinterwordspacing}{\spaceskip=\fontdimen2\font plus
\BIBentryALTinterwordstretchfactor\fontdimen3\font minus
  \fontdimen4\font\relax}
\providecommand{\BIBforeignlanguage}[2]{{%
\expandafter\ifx\csname l@#1\endcsname\relax
\typeout{** WARNING: IEEEtran.bst: No hyphenation pattern has been}%
\typeout{** loaded for the language `#1'. Using the pattern for}%
\typeout{** the default language instead.}%
\else
\language=\csname l@#1\endcsname
\fi
#2}}
\providecommand{\BIBdecl}{\relax}
\BIBdecl

\bibitem{377}
D.~Gesbert, M.~Kountouris, R.~W. Heath~Jr., C.~B. Chae, and T.~Salzer, ``From
  single user to multiuser communications: Shifting the {MIMO} paradigm,''
  \emph{{IEEE} Signal Process. Mag.}, vol.~24, no.~5, pp. 36--46, Oct. 2007.

\bibitem{373}
T.~L. Marzetta, ``Noncooperative cellular wireless with unlimited numbers of
  base station antennas,'' \emph{{IEEE} Trans. Wireless Commun.}, vol.~9,
  no.~11, pp. 3590--3600, Nov. 2010.

\bibitem{459}
H.~Q. Ngo, E.~G. Larsson, and T.~L. Marzetta, ``Energy and spectral efficiency
  of very large multiuser {MIMO} systems,'' \emph{{IEEE} Trans. Commun.},
  vol.~61, no.~4, pp. 1436--1449, Apr. 2013.

\bibitem{37}
G.~J. Foschini and M.~J. Gans, ``On limits of wireless communications in a
  fading environment when using multiple antennas,'' \emph{Wireless Personal
  Communications}, vol.~6, no.~3, pp. 311--355, Mar 1998.

\bibitem{36}
E.~Telatar, ``Capacity of multi-antenna {G}aussian channels,'' \emph{European
  transactions on telecommunications}, vol.~10, no.~6, pp. 585--596, Nov. 1999.

\bibitem{212}
H.~Weingarten, Y.~Steinberg, and S.~Shamai~(Shitz), ``The capacity region of
  the {Gaussian} multiple-input multiple-output broadcast channel,''
  \emph{{IEEE} Trans. Inf. Theory}, vol.~52, no.~9, pp. 3936--3964, Sep. 2006.

\bibitem{263}
G.~Caire and S.~Shamai, ``On the achievable throughput of a multiantenna
  {G}aussian broadcast channel,'' \emph{{IEEE} Trans. Inf. Theory}, vol.~49,
  no.~7, pp. 1691--1706, Jul. 2003.

\bibitem{310}
N.~Jindal and A.~Vishwanath, S.~Goldsmith, ``On the duality of {G}aussian
  multiple-access and broadcast channels,'' \emph{{IEEE} Trans. Inf. Theory},
  vol.~50, no.~5, pp. 768--783, May 2004.

\bibitem{453}
S.~Vishwanath, N.~Jindal, and A.~Goldsmith, ``Duality, achievable rates, and
  sum-rate capacity of {G}aussian {MIMO} broadcast channels,'' \emph{{IEEE}
  Trans. Inf. Theory}, vol.~49, no.~10, pp. 2658--2668, Oct. 2003.

\bibitem{312}
L.~Zhang, R.~Zhang, Y.~C. Liang, Y.~Xin, and H.~V. Poor, ``On {G}aussian {MIMO
  BC-MAC} duality with multiple transmit covariance constraints,'' \emph{{IEEE}
  Trans. Inf. Theory}, vol.~58, no.~4, pp. 2064--2078, Apr. 2012.

\bibitem{271}
P.~Viswanath and D.~N.~C. Tse, ``Sum capacity of the vector {G}aussian
  broadcast channel and uplink-downlink duality,'' \emph{{IEEE} Trans. Inf.
  Theory}, vol.~49, no.~8, pp. 1912--1921, Aug. 2003.

\bibitem{5003}
{M. Costa}, ``Writing on dirty paper,'' \emph{{IEEE} Trans. Inf. Theory},
  vol.~29, no.~3, pp. 439--441, May 1983.

\bibitem{182}
Q.~H. Spencer and A.~L. Swindlehurst, ``Zero-forcing methods for downlink
  spatial multiplexing in multiuser {MIMO} channels,'' \emph{{IEEE} Trans.
  Signal Process.}, vol.~52, no.~2, pp. 461 -- 471, Feb. 2004.

\bibitem{226}
W.~Yu and T.~Lan, ``Transmitter optimization for the multi-antenna downlink
  with per-antenna power constraints,'' \emph{{IEEE} Trans. Signal Process.},
  vol.~55, no.~6, pp. 2646--2660, Jun. 2007.

\bibitem{5004}
{E. Bjornson, J. Hoydis, and L. Sanguinetti}, ``Massive {MIMO} has unlimited
  capacity,'' \emph{{IEEE} Trans. Wireless Commun.}, vol.~17, no.~1, pp.
  574--590, Jan. 2018.

\bibitem{2005}
{Y. Zeng, X. Xu, S. Jin, and R. Zhang}, ``Simultaneous navigation and radio
  mapping for cellular-connected uav with deep reinforcement learning,''
  \emph{IEEE Trans. Wireless Commun.}, vol.~20, no.~7, pp. 4205--4220, 2021.

\bibitem{638}
W.~Zhao, M.~Ammar, and E.~Zegura, ``A message ferrying approach for data
  delivery in sparse mobile ad hoc networks,'' \emph{In Proc. ACM Mobihoc}, May
  2004.

\bibitem{653}
S.~Jain, K.~Fall, and R.~Patra, ``Routing in a delay tolerant network,''
  \emph{{Proc. ACM SIGCOMM}}, pp. 1--13, Jan. 2004.

\bibitem{641}
Y.~Zeng, R.~Zhang, and T.~J. Lim, ``Throughput maximization for {UAV}-enabled
  mobile relaying systems,'' \emph{{IEEE} Trans. Commun.}, vol.~64, no.~12, pp.
  4983--4996, Dec. 2016.

\bibitem{1095}
{Y. Zeng, Q. Wu, and R. Zhang}, ``Accessing from the sky: a tutorial on {UAV}
  communications for {5G} and beyond,'' \emph{Proc. of the IEEE}, vol. 107,
  no.~12, pp. 2327--2375, Dec. 2019.

\bibitem{904}
{Y. Zeng and R. Zhang}, ``{Energy-efficient UAV communication with trajectory
  optimization},'' \emph{{IEEE Trans. Wireless Commun.}}, vol.~16, no.~6, pp.
  3747--3760, Jun. 2017.

\bibitem{624}
F.~Bohagen, P.~Orten, and G.~E. Oien, ``Design of optimal high-rank
  line-of-sight {MIMO} channels,'' \emph{{IEEE} Trans. Wireless Commun.},
  vol.~6, no.~4, pp. 1420--1425, Apr. 2007.

\bibitem{5005}
{H. Lu, Y. Zeng, C. You, et al}, ``A tutorial on near-field {XL-MIMO}
  communications toward {6G},'' \emph{{IEEE Commun. Surveys Tutorials}},
  vol.~26, no.~4, pp. 2213--2257, 4th quarter 2024.

\bibitem{5006}
{W. K. New, K.-K. Wong, H. Xu, et al}, ``A tutorial on fluid antenna system for
  {6G} networks: Encompassing communication theory, optimization methods and
  hardware designs,'' \emph{{IEEE Commun. Surveys Tutorials}}, vol.~27, no.~4,
  pp. 2325--2377, Aug. 2025.

\bibitem{5007}
{L. Zhu, W. Ma, W. Mei, et al}, ``A tutorial on movable antennas for wireless
  networks,'' \emph{{IEEE Commun. Surveys Tutorials}}, vol.~28, pp. 3002--3054,
  Feb. 2025.

\bibitem{649}
Y.~Zeng, R.~Zhang, and T.~J. Lim, ``Wireless communications with unmanned
  aerial vehicles: opportunities and challenges,'' \emph{IEEE Commun. Mag.},
  vol.~54, no.~5, pp. 36--42, May 2016.

\bibitem{5009}
{M. Mozaffari, W. Saad, M. Bennis, and M. Debbah}, ``Communications and control
  for wireless drone-based antenna array,'' \emph{{IEEE Trans. Commun.}},
  vol.~67, no.~1, pp. 820--834, Jan 2019.

\bibitem{5011}
{H. Lu, Y. Zeng, S. Ma, B. Li, S. Jin, and R. Zhang}, ``{Wireless Communication
  for Low-altitude Economy with UAV Swarm Enabled Two-Level Movable Antenna
  System},'' \emph{to appear in IEEE Trans. Wireless Commun.,
  arxiv.org/abs/2505.22286}, 2026.

\bibitem{5008}
{N. Gao, X. Li, S. Jin, and M. Matthaiou}, ``{3-D Deployment of UAV Swarm for
  Massive MIMO Communications},'' \emph{{IEEE} J. Sel. Areas Commun.}, vol.~39,
  no.~10, pp. 3022--3034, Oct. 2021.

\bibitem{5010}
{S. Hanna, E. Krijestorac, and D. Cabric}, ``{UAV Swarm Position Optimization
  for High Capacity MIMO Backhaul},'' \emph{{IEEE} J. Sel. Areas Commun.},
  vol.~39, no.~10, pp. 3006--3021, Oct. 2021.

\bibitem{227}
M.~Grant and S.~Boyd, \emph{{CVX}: Matlab software for disciplined convex
  programming, version 2.1}, available online at http://cvxr.com/cvx.

\bibitem{2101}
Y.~Zeng and X.~Xu, ``Toward environment-aware {6G} communications via channel
  knowledge map,'' \emph{IEEE Wireless Commun.}, vol.~28, no.~3, pp. 84--91,
  2021.

\end{thebibliography}

\end{document}